\documentclass[11pt,a4paper]{article}
\usepackage[labelfont=bf]{caption}

\usepackage{xr-hyper} 
\usepackage{amsmath,amsfonts,amssymb,amsthm,booktabs,color,epsfig,graphicx,hyperref,url}
\usepackage{subcaption,verbatim,dsfont,array,relsize,url,siunitx,tikz,mathtools,pgf,algorithm,tikz-3dplot}

\hypersetup{
	colorlinks,
	linkcolor={blue!80!black},
	citecolor={blue!80!black},
	urlcolor={blue!80!black}
}

\usetikzlibrary{positioning}
\usetikzlibrary{calc,tikzmark}
\usetikzlibrary{arrows,automata}
\usepackage[english]{babel}
\usepackage[noend]{algpseudocode}
\usepackage[titletoc]{appendix}
\usepackage[authoryear]{natbib}
\numberwithin{equation}{section}
\newtheorem{theorem}{Theorem}
\newtheorem{lemma}{Lemma}
\newtheorem{remark}{Remark}
\newtheorem{example}{Example}
\newtheorem{proposition}{Proposition}
\newtheorem{definition}{Definition}

\newtheorem{corollary}{Corollary}
\newtheorem{fig}{Figure}

\newcommand{\bthe}{\begin{theorem}}
\newcommand{\ethe}{\end{theorem}}

\newcommand{\ben}{\begin{enumerate}}
\newcommand{\een}{\end{enumerate}}

\newcommand{\bit}{\begin{itemize}}
\newcommand{\eit}{\end{itemize}}

\newcommand{\beq}{\begin{equation}}
\newcommand{\eeq}{\end{equation}}

\newcommand{\ble}{\begin{lemma}}
\newcommand{\ele}{\end{lemma}}

\newcommand{\bde}{\begin{definition}}
\newcommand{\ede}{\end{definition}}

\newcommand{\bco}{\begin{corollary}}
\newcommand{\eco}{\end{corollary}}

\newcommand{\bpr}{\begin{proposition}}
\newcommand{\epr}{\end{proposition}}

\newcommand{\brem}{\begin{remark}}
\newcommand{\erem}{\end{remark}}

\newcommand{\bproof}{\begin{proof}}
\newcommand{\eproof}{\end{proof}}

\newcommand{\bexam}{\begin{example}}
\newcommand{\eexam}{\end{example}}

\newcommand{\bfi}{\begin{fig}}
\newcommand{\efi}{\end{fig}}

\newcommand{\btab}{\begin{tab}}
\newcommand{\etab}{\end{tab}}

\newcommand{\beao}{\begin{eqnarray*}}
\newcommand{\eeao}{\end{eqnarray*}\noindent}

\newcommand{\beam}{\begin{eqnarray}}
\newcommand{\eeam}{\end{eqnarray}\noindent}

\newcommand{\barr}{\begin{array}}
\newcommand{\earr}{\end{array}}

\newcommand{\bdis}{\begin{displaymath}}
\newcommand{\edis}{\end{displaymath}\noindent}

\def\E{{\mathbb E}}
\def\R{{\mathbb R}}

\def\cals_+{{\cals_+}}

\def\cals{{\mathcal{S}}}

\newcommand{\bs}{\boldsymbol}

\newcommand{\bsx}{\boldsymbol{X}}

\newcommand{\stp}{\stackrel{P}{\rightarrow}}

\newcommand{\stD}{\stackrel{\mathcal{D}}{\rightarrow}}

\newcommand{\stw}{\stackrel{ w}{\rightarrow}}

\newcommand{\al}{{\alpha}}

\newcommand{\eps}{\varepsilon}

\newcommand{\DAG}{{\rm DAG}}

\newcommand{\an}{{\rm an}}
\newcommand{\pa}{{\rm pa}}

\newcommand{\des}{{\rm de}}
\newcommand{\An}{{\rm An}}
\newcommand{\Pa}{{\rm Pa}}

\newcommand{\Des}{{\rm De}}

\let\norm\undefined 
\DeclarePairedDelimiter\norm{\lVert}{\rVert}

\def\D{\mathcal{D}}


\title{\vspace{-.7cm}Causal discovery in heavy-tailed linear structural equation models via scalings}
\author{Mario Krali\\Institute of Mathematics\\ 
	\'Ecole Polytechnique F\'ed\'erale de Lausanne (EPFL)\\
	e-mail:  \texttt{mario.krali@epfl.ch}}
\begin{document}
\maketitle

\begin{abstract}
	Causal dependence modelling of multivariate extremes is intended to improve our understanding of the relationships amongst variables associated with rare events. Regular variation provides a standard framework in the study of extremes. This paper concerns the linear structural equation model with regularly varying noise variables. 
	We focus on extreme observations generated from such a model and propose a causal discovery method based on the scaling parameters of its extremal angular measure. We implement the method as an algorithm, establish its consistency and evaluate it by simulation and by application to river discharge datasets. Comparison with the only alternative extremal method for such model reveals its competitive performance.
\\

\noindent
{\em AMS 2020 Subject Classifications:}  primary:
\,\,\,60G70; 
\,\,\,62D20; 
\,\,\,62G32;  
\,\,\,62H22 

\noindent
\textbf{Keywords:} causal discovery;  extreme value theory; graphical model; heavy tails; linear structural equation model; multivariate regular variation	
\end{abstract}

\section{Introduction and motivation}

Extreme events can result in catastrophic consequences, of which floods, droughts and heatwaves are obvious manifestations. Extreme value theory (EVT) has become the standard tool in studying such phenomena and in helping to better prepare to mitigate their effects. Often, these events are subject to cause-and-effect types of behaviour, which raises questions on
the underlying extremal causal mechanisms.

Causal discovery and structure learning assist in identifying causal structures and serve as a stepping-stone towards causal inference, and are studied using graphical models \citep{pearl, lau}. 
The process of uncovering cause-and-effect relationships involves a large number of variables, thus giving rise to high-dimensional settings. Observations in high dimensions are notorious for posing challenges in statistics, and this is further exacerbated in the study of extremes, which manifest only rarely \citep[Ch. 7]{DHF}. 

The focus of this paper is on graphical modelling, and in particular on the linear structural equation model (LSEM) \citep{pearl}. An LSEM is supported on a directed acyclic graph (DAG) $\mathcal{D}=(V, E)$ with nodes $V=\{1,\ldots, d\}$ and edge set $E\subset V\smash{\times V}$, and is recursively defined as 
\begin{align}\label{p2lsem}
X_i=\sum_{j\in\textrm{pa}(i)}c_{ij}X_j+s_{ii}Z_i,\hspace{1cm}i\in V,
\end{align}
for innovation variables $Z_i$. We shall restrict ourselves to  non-negative edge weights, so $c_{ij}\geq0$ for $i\neq j$, and innovation weights $s_{ii}>0$ in \eqref{p2lsem}, and refer to the resulting model as LSEM+. {This restriction arises because of the proposed methodology, but one may nevertheless find it plausible given our focus on multivariate extremes, which are often associated with the behaviour of maxima.} In EVT, dependence between extreme observations of the variables is characterised by the limiting angular measure.

The LSEM+ is closely related to the so-called recursive max-linear model (RMLM) studied in \cite{gk}, which is known to possess a discrete limiting angular measure \citep{foug}. Unlike the RMLM, which only allows the extreme shocks to travel through the network, the LSEM+  preserves smaller shocks. However, under the framework of regularly varying innovations, the LSEM+ also has a discrete limiting angular measure.  The RMLM manifests multivariate extremes in a ray-like fashion starting from a penultimate or pre-asymptotic level, but by contrast, the LSEM+  only shows such behaviour in the limit. The finite number of data points in real datasets naturally corresponds to a penultimate setup, which makes the LSEM+ 
more realistic in applications.

Recent research has focused on combining extremes and methods from graphical modelling \citep{gk, engelke:hitz:18, eglearn}, causal discovery \citep{gnecco, KK, KDK, mhalla, TBK},  causality for extremes of time-series \citep{bodik}, sparsity modelling \citep{ES,cooley, meyerclust, goix, leecooley}), extreme quantile prediction \citep{exforest} and clustering \citep{Chautru, JanWan}. Another prominent facet of extreme statistics concerns conditional modelling of extremes \citep{hefftawn}. 
Although each of these methods enriches a specific area of statistical modelling, questions remain on how to bridge the resulting conclusions to other aspects of statistical analysis. 
For instance, conditional extremes, seemingly a well-established sub-domain, is intended to capture the phenomenon of asymptotic independence between extremes, but does not offer a viable way to incorporate the resulting conditional marginal distributions into an overall joint dependence structure. Another example is the undirected graphical structure in \cite{engelke:hitz:18}, which makes it impossible to infer causal relationships amongst the node variables.

{The choice of the LSEM+ is partly motivated by the fact that discrete limiting angular measures form a dense subset among all angular measures \citep{foug}. The idea is thus to work with a model that is able to encode extreme causal dependencies and that may potentially approximate the extremal dependence structure. Secondly, the equations of the LSEM+ are instances of linear regression, which find wide applications in many domains. 
	Apart from applications in causal discovery, extremal causal dependencies may also complement the conditional approach of~\citet{hefftawn}, providing possibly useful information in the selection of a conditioning variable based on causal dependencies.}

{Throughout our analysis we work with the standardised representation of the LSEM+~\eqref{p2lsem} in order to facilitate the study of its extremal causal dependence structure; however, because of our focus on extremes, the proposed methodology can cover a wider class of models whose extremal behaviour is asymptotically equivalent to the LSEM+. One can think here of a model $\bs Y= \bs X+ \bs V$, with the components of $\bs V$ having lighter tail indices relative to the components of the LSEM+ $\bs X$. As the vectors $\bs X$ and $\bs V$ need not have the same dependence structure, the extreme causal relations of $\bs Y$ may differ from those in the bulk~of~its~distribution.} 
 
In this paper we propose a  methodology to uncover causal dependencies between node variables using only extreme observations. Under the standard assumption of regular variation, we exploit certain asymmetries reflected in the scaling parameters of the extreme angular measure. In particular, we refine the methodology of \cite{KK}, which relies on the angular measure of the $d$-dimensional vector $\bs X$ to compute the scaling parameters. In contrast, our approach starts with the bivariate angular measure, and then augments its dimension by accounting for the node variables identified.
\cite{gnecco} develop an algorithm for the causal ordering of heavy-tailed LSEMs using the causal tail coefficient. We use their approach as a benchmark in a simulation study and data application. Building on \cite{gnecco}, and considering settings with confounders, \citet{PCD} propose testing for direct causal links between two variables by including confounders in a regression setting. Several methods have been proposed to tackle structure learning for the LSEM in the non-extremal setting, using all $n$ observations. Amongst these are LiNGAM \citep{lingam1}, the pairwise likelihood ratio method \citep{pairwiselingam}, DirectLiNGAM \citep{directlingam},  and modified DirectLiNGAM \citep{wangcausal}. \citet{drtonreview} survey these methods.

This paper is organised as follows. Section \ref{p2prelims} introduces the necessary basics from graphical modelling and gives some properties of the LSEM+. Section \ref{p2sec:RV} outlines the required notions from multivariate regular variation and scalings. In Section \ref{p2slearn} we introduce the scaling methodology used for causal discovery and show its consistency. Section \ref{p2simdata} describes a simulation study and outlines applications to river discharges from stations along the Upper Danube basin and from the Rhine basin in Switzerland. Appendix \ref{p2sec:ARV} provides the necessary regular variation background; Appendix \ref{p2proofs} contains all proofs; Appendix \ref{p2estprocedure} outlines the estimation procedure and gives the consistency of the estimated causal orderings; and Appendix \ref{p2bplot5000} contains further numerical results.

\section{Graphs and linear structural equation models}\label{p2prelims}
\subsection{Vocabulary and notation}
We first introduce terminology for directed graphs \citep{lau}.
Let $\mathcal{D}=(V, E)$ denote a directed acyclic graph (DAG) with node set $V=\{1,\dots,{d}\}$ and edge set $E\subset V\times V$. For each node $i\in V$ the sets $\pa(i)=\{j\in V: (j,i)\in E\}$, $\an(i)$ and $\des(i)$ respectively denote its parents, ancestors and descendants, and we write $\Pa{(i)}=\pa(i)\cup \{i\} $, $\An(i)=\an(i)\cup \{i\}$ and $\Des(i)=\des(i)\cup\{i\}$. 
Finally, for $I\subset V$, $\an(I)$ denotes the ancestors of all nodes in $I$, and we write $\An(I)=\an(I)\cup I$.
A node $i\in V$ is called a source node if $\pa(i)=\emptyset$, and we let $V_0$ denote  the set of all source nodes.
{We use  $j\to i$ to denote an edge from node $j$ to node $i$}. A path $p_{ji}\coloneqq[\ell_0=j\to \ell_1 \to\cdots\to \ell_{m}=i]$ from $j$ to $i$ has length $|p_{ji}|=m$, and the set of all such paths is denoted by $P_{ji}$. 
We also write $j\rightsquigarrow i$ for $p_{ji}$. 

Structural equation models (SEMs) supported on graphs provide a convenient way to capture causal dependencies between observed random variables.
The underlying graph associates each node $j$ with a random variable $X_j$, and we  
say that $X_j$ causes $X_i$ (or $j$ causes $i$) whenever there is a path $j \rightsquigarrow i$. Our definition of causal dependence is thus path-induced and founded on ancestral relations, rather than solely based on the presence of edges.

A DAG $\mathcal{D}=(V,E)$ is {well-ordered} if  $i<j$ for all $j\in \pa(i)$, which we call a causal order.

Given nodes $i,j,k\in V$, we say that $X_i$ is a confounder of $X_j$ and $X_k$ (or $i$ is a confounder of $j$ and $k$) if there exist distinct paths $i \rightsquigarrow j$ and $i\rightsquigarrow k$ which do not pass through $k$ and $ j$, respectively.

\subsection{Solution of linear structural equation models  via innovations }

The LSEM $\boldsymbol{X}$ in \eqref{p2lsem} has a solution in terms of the innovation vector $\boldsymbol{Z}$. Given a well-ordered DAG $\D=(V,E)$, we take the edge coefficient matrix $C$ to be strictly upper-triangular. Let the diagonal matrix $S$ have entries $s_{ii}$ serving as weights for the components of $\bs Z$. The matrix $C$ is nilpotent, so  matrix inversion gives
\begin{align}\label{p2Amat}
\boldsymbol{X}=(I-C)^{-1}S\boldsymbol{Z}=A\boldsymbol{Z},
\end{align}
where we call $A$ the (innovation) coefficient matrix.
This implies that
\begin{align} \label{p2innovrep}
X_i=\sum_{j\in \textrm{An}(i)}a_{ij}Z_j, \quad i\in V,
\end{align} 
since $a_{ij}=0$ for $j\notin\textrm{An}(i)$. 

In comparison to the LSEM in \eqref{p2lsem}, the max-linear analogue is defined via
\begin{align}\label{p2rmlm}
X_i=\bigvee_{j\in\textrm{pa}(i)}c_{ij}X_j\vee s_{ii}Z_i,\hspace{1cm}i\in V,
\end{align}
with non-negative coefficients $c_{ij}$ and with $\vee$ denoting the maximum operator.
A path analysis approach analogous to max-linear operations \citep{gk} shows that the coefficients $a_{ij}$ of the LSEM equal
\begin{align}\label{p2lsemcoef}
\hspace{-.3cm}a_{ij}=\underset{p_{ji}\in P_{ji}}{\sum}d(p_{ji}) \mbox{ for } j\in {\an}(i),\quad a_{ij}=0 \mbox{ for }  j\in V\setminus {\An}(i),\quad a_{ii}=s_{ii},
\end{align}
with path weights $d(p_{ji})\coloneqq s_{jj}c_{k_1j}\cdots c_{ik_{l-1}}$ corresponding to the product of the edge weights along the path $p_{ji}$. The coefficient $a_{ij}$ therefore equals the sum of all path weights from node $j$ to node $i$, unlike the RMLM, where $a_{ij}$ is computed only for the max-weighted path.


\subsection{Linear structural equation models with positive edge weights}

While non-negative edge-weights constrain the model, positive dependence is of particular interest when studying extremes, which manifest in the behaviour of maxima and, 
 in general, are supported on the positive orthant $\mathbb{R}_+^d$. As seen in \eqref{p2lsemcoef}, non-negativity of the edge weights translates into non-negative path weights. By contrast, a negative edge coefficient $c_{ij}$ might lead to a negative  coefficient $a_{ij}$, which moves the support of $\bs X$ outside $\mathbb{R}_+^d$.

The study of maxima and peaks over thresholds in applications typically involves a pre-processing step to recenter and truncate the observations to the non-negative orthant or uses transformations of linear models \citep{cooley, leecooley}. These papers focus only on extremal dependence on the positive orthant, corresponding to setting all negative entries of a general matrix $A\in\mathbb{R}^{d\times d}$ in \eqref{p2Amat} to zero. Proposition \ref{p2scalee} shows that the matrix $A$ provides information on the support of the angular measure, which fully characterises the dependence structure of multivariate extremes. The non-negativity constraint on the matrix $C$ therefore serves as a natural remedy that ensures non-negativity of $A$ in \eqref{p2Amat}.  

We now study some properties of the innovation coefficient matrix $A$, which resemble certain characteristics of the coefficient matrix in the RMLM. The next proposition characterises whether paths $j\rightsquigarrow i$ between nodes $i$ and $j$  are of the form $j\rightsquigarrow k\rightsquigarrow i$ for some $k\in \des(j)\cap \an(i)$, which proves to be important in deriving properties of the  diagonal coefficients $a_{ii}$ of the coefficient matrix.

\begin{proposition}\label{p2allpathcond}
	Consider an LSEM+ supported on the DAG $\mathcal{D}=(V,E)$, with edge weight matrix $C\in\mathbb{R}_+^{d\times d}$ and innovation coefficient matrix $A\in\mathbb{R}_+^{d\times d}$. For $i\in V$ and $j\in \an(i)$, all paths $j\rightsquigarrow i$ are of the form $j\rightsquigarrow k\rightsquigarrow i$, where $k\in\des(j)\cap\an(i)$, if and only if 
	\begin{align}\label{p2pathdecomp}
	a_{ij}=\frac{a_{ik}a_{kj}}{a_{kk}};
	\end{align}
	otherwise $a_{ij}>{a_{ik}a_{kj}}/ a_{kk}$.
\end{proposition}	

The next corollary  is a consequence of Proposition \ref{p2allpathcond} and is analogous to Corollary~3.12 of \cite{gk}.
\begin{corollary}\label{p2cor1}
In the setting of Proposition \ref{p2allpathcond}, $a_{ij}\geq {a_{ik}a_{kj}}/{a_{kk}}$ for any nodes $i,j,k$ of $V$.
\end{corollary}

We now consider the standardised innovation coefficient matrix with entries defined as $(\bar{A})_{ij}=\big({a_{ij}^\alpha}/{\sum_{k=1}^d{a_{ik}^\alpha}}\big)^{1/\alpha}$ for $\alpha>0$, which is used in deriving the extremal scalings  in Lemma \ref{p2scalcoll}.
The next lemma shows that  the diagonal entries of $\bar{A}$ are the largest in their respective columns, a property that also holds for the coefficient matrix of an RMLM and is important in establishing asymmetries in the scaling methodology in Theorems \ref{p2sourcenodes} and \ref{p2descnodes}.

\begin{lemma}\label{p2ineq}
	Let $\bar{A}\in \mathbb{R}_+^{d\times d}$ be the standardised innovation coefficient matrix of an LSEM+. Then $\bar{a}_{jj}>\bar{a}_{ij}$ whenever $i\neq j$.
\end{lemma}

\section{Multivariate regular variation} \label{p2sec:RV}

As our interest lies in the extremal behaviour of the LSEM+, we let the innovation vector $\boldsymbol{Z}\in \R_+^d$ be regularly varying with index $\alpha>0$, i.e., $\boldsymbol{Z}\in{\rm RV}_+^{d}(\alpha)$,  and with independent and standardised components, i.e., for all $i\in\{1,\dots,{d}\}$, $n\mathds{P}(n^{-1/\alpha} Z_i>z)\to z^{-\alpha}$ for  $z>0$ as $n\to\infty$ .
The LSEM+ has properties  similar to max-linear models 
with independent regularly varying innovations \citep{Wang2011, gk, KK}.
Both are multivariate regularly varying and possess a discrete angular measure $H_{\bsx}$ on the positive unit simplex $\Theta_+^{{ {d}}-1}=\{\boldsymbol{\omega}\in \mathbb{R}^{ {d}}_+: \norm{\boldsymbol{\omega}}=1\}$. The measure	$H_{\boldsymbol{X}}$ is not necessarily a probability measure but can be normalised into one, $\tilde H_{\bsx}$, following Remark \ref{p2specmass2}.
Details on multivariate regular variation, the angular representation and the angular measure are given in Appendix~\ref{p2sec:ARV}; see also \citet{sres,ResnickHeavy}.

Our methodology is based on a particular extremal dependence measure introduced in Propositions~3 and~4 of \cite{lars}; see also \citet[\S~4]{cooley} and \citet[\S~2.2]{KK}. 

\bde\label{p2scaledef}
Let $\bsx\in {\rm RV}^d_+(2)$ have angular measure $H_{\bs X}$ on $ \Theta_+^{{ {d}}-1}$  and angular representation $(R,\boldsymbol{\omega})=(\norm{\bsx}, {\bs X}/{R})$. 
For $ i,j\in \{1,\ldots, {d}\}$ define the extremal  dependence measure 
\begin{align*}
\sigma_{ij}^2 = \sigma_{X_{ij}}^2&\coloneqq\int_{\Theta_+^{{ {d}}-1}}\omega_i \omega_j {\rm d}H_{\boldsymbol{X}}(\boldsymbol{\omega}), \quad \boldsymbol{\omega}=(\omega_1,\dots,\omega_{ {d}})\in \Theta_+^{{ {d}}-1}.
\end{align*}
We call $\sigma_i \coloneqq\sigma_{{X}_{ii}}$ the scaling parameter, or the scaling,  of $X_i$. 
\ede

The next proposition is similar to results that characterise the angular measure $H_{\bs X}$ from \citet{foug} and \citet{cooley}, adapted to account for linear operations. In this paper we work with the regular variation index $\al=2$ and the Euclidean norm $\|\cdot\|$. In case of the LSEM+~\eqref{p2lsem}, this allows a convenient representation of both the angular measure and the scalings from the entries of the coefficient matrix $A$. 

\bpr\label{p2scalee}
Consider the LSEM+ vector ${\boldsymbol{X}}=A\boldsymbol{Z}$, where $\boldsymbol{Z}\in {\rm RV}^{d}_+(2)$ has standardised margins and $A\in \mathbb{R}^{{d}\times{d}}_+$. Then
\begin{enumerate}
	\item[(i)]
	the angular measure $H_{\bsx}$ is supported on $(\bs a_{k}/\norm{\bs a_{k}}){=(\bs a_k/{(\sum_{i=1}^{d} a_{ik}^2)^{1/2}})}$ for  $k\in V$;
	\item[(ii)]
	$\sigma_{ij}^2=(AA^\top)_{ij}=\sum_{k=1}^da_{ik} a_{jk}$ for $i,j\in V$; and
	\item[(iii)]
	each margin $X_i$ has squared scaling $\sigma^2_i=\sum_{k=1}^{d} a_{ik}^2$  for $i\in V$.
\end{enumerate}
\epr

As is common in extremes, we only work with the standardised LSEM+ vector $\bsx$, obtained from~\eqref{p2lsem} by standardizing the innovation coefficient matrix $A$. The discussion preceding Lemma \ref{p2ineq} then implies that we are working directly with $\bar{A}$, i.e., $\boldsymbol{X}=\bar{A}\boldsymbol{Z}$, whereby Proposition~\ref{p2scalee} entails unit scalings for the margins.


We summarise the assumptions used throughout the rest of the paper.\vspace{2mm}\\
{\bf {Assumptions A:}}
\begin{enumerate}\label{p2modelassump}
	\item[(A1)]
	The innovations vector $\boldsymbol{Z}\in {\rm RV}^{d}_+(2)$ has independent and standardised components.
	\item[(A2)]
	The choice of norm is the Euclidean norm, denoted by $\|\cdot\|$.
	\item[(A3)]
	The components of $\bsx$ are standardised. 
\end{enumerate}

In preparation for Theorems \ref{p2sourcenodes} and \ref{p2descnodes} and the examples of the next section, we briefly summarise some important properties of the scalings. In particular, we make use of maxima  over the rescaled components of the vector $\boldsymbol{X}$ under Assumptions A. Lemma \ref{p2scalcoll}, which is a consequence of the discussion in Appendix \ref{p2scalings}, adapts Lemma 6 of \cite{KK} to the LSEM+ and employs Lemma~\ref{p2ineq} to characterise the scalings of the maxima $M_I=\max(X_i: i\in I)$, defined as
\begin{align*}
\sigma_{M_{\boldsymbol{I}}}^2 =\int_{\Theta_+^{{d}-1}}\underset{i\in I}{\vee}\omega_i^{2} {\rm d}H_{\boldsymbol{X}}(\boldsymbol{\omega}),
\end{align*}
in terms of the coefficient matrix $A$.

\ble[\citet{KK}]\label{p2scalcoll}
Let $\boldsymbol{X}$ be an LSEM+ satisfying Assumptions A. Then the random variable $M_{I}\coloneqq \max(X_i, i\in I)$ lies in ${\rm RV}_+(2)$ and its squared scaling equals\\
{\rm{(i)}} $\sigma_{M_{I}}^2 =\sum_{k=1}^{d}\bigvee_{i\in I} a_{ik}^2$ if $I\subsetneq V$, {then}; and\\
{\rm{(ii)}} $\sigma_{M_{I}}^2 = \sum_{k=1}^{d}  a_{kk}^2$ if $I =V$.
\ele

\section{Structure learning }\label{p2slearn}
Structure learning, or causal discovery, often relies on specific assumptions on the graphical structure supporting the random variables. These typically amount to assuming causal sufficiency  \citep[p. 45]{spirt}, which is closely related to Markovianity of the model \citep[p. 30]{pearl}, or to no `hidden confounding', which ensures that there are no unmeasured sources of error \citep[p. 62]{pearl}. 
While Markovianity may not hold in practice, causing it to attract much criticism \cite[p. 252]{pearl}, it is nevertheless exploited by many methods, including the graphical lasso \citep{glasso}, which relies on conditional independence properties of the multivariate Gaussian distribution. If there are no hidden confounders in an LSEM+, we can uncover its causal structure recursively.  

\subsection{Identification of source nodes} \label{p2sources}
We first  identify the source nodes by exploiting scaling asymmetries in pairwise maxima of the node variables. Our approach differs from that of \cite{KK}, who consider maxima over the $d$ node variables of the DAG. Instead, we focus on maxima of lower-dimensional subvectors, which  are important in determining the angular measure: the approach of \citet{KK} uses the $d$-dimensional angular measure of $\bs X$, whereas the approach we propose relies on the angular measure of the corresponding subvector. We address this difference in the following remark, which discusses some disadvantages of high-dimensional angular measures.

\begin{remark}\label{p2dimtrick}
Consider the LSEM+ vector  $\bs X\in{\rm RV}^d_+(2)$ and its subvector $\bs X_{ij}=(X_i, X_j)$ with normalised angular measures $\tilde H_{\bs X}$ and $\tilde H_{\bs X_{ij}}$, respectively supported on $\Theta^{d-1}_+$ and $\Theta_+$.  Proposition~\ref{p2scalee}~(i) shows that the columns of the coefficient matrix $A$ provide the support for $\tilde H_{\bs X}$, and that in combination with representation \eqref{p2innovrep} we can index these columns according to the nodes. Similarly, we can combine the  $i$-th and $j$-th rows of $A$ into the matrix $A^{(ij)}\in \mathbb{R}^{2\times d}_+$, whose columns $\bs a^{(ij)}\in\mathbb{R}^{2}_+$ provide the support for $\tilde H_{\bs X_{ij}}$.

Suppose that $\bs a_j$ has entries such that $a_{jj}>0$ and $a_{ij}=0$ for $i\neq j$, which implies that $j\notin \an(i)$. Typical examples of such nodes $j$ are the descendants of a lower order on a DAG. Consider the two measures $\tilde H_{\bs X}$ and $\tilde H_{\bs X_{ij}}$:
\item[(i)] under $\tilde H_{\bs X}$, we invoke Proposition \ref{p2discspect} and find that the probability of observing realisations supported on the atom $\bs a_j/\norm{\bs a_j}$ equals $\norm{\bs a_j}/d=a_{jj}/d$; 
\item[(ii)] under $\tilde H_{\bs X_{ij}}$, the probability of observing realisations supported on the atom $\bs a^{(ij)}_j/\norm{\bs a^{(ij)}_j}$ equals $\norm{\bs a^{(ij)}_j}/2=a_{jj}/2$.
\end{remark}

Remark \ref{p2dimtrick} illustrates that the angular measure $\tilde H_{\bs X_{ij}}$ may be preferable to $\tilde H_{\bs X}$ for a high dimension $d$, as the former assigns higher probability to the atoms $\bs a^{(ij)}_j$, whose entries are important in the computation of the scalings.  Similar reasoning applies when there are very small, but positive, entries $a_{ij}$. We show in Theorem \ref{p2sourcenodes} that deriving the scalings from the measures $\tilde H_{\bs X_{ij}}$ over all pairs $i\neq j$ carries necessary and sufficient information to identify the source nodes, which implies that there is no loss of information from not using $\tilde H_{\bs X}$ as in Theorem 2 of \citet{KK}.  
However, the main drawback of using $\tilde H_{\bs X_{ij}}$ is the cost of computing scalings over the available $d(d-1)$ pairs in the DAG, whereas the approach of \cite{KK} only requires computations of scalings over $d$ maxima.

We use the following notation for partially rescaled maxima:
for $a>1$ we write
\begin{equation}\label{p2Mmam}
M_{i,aj} \coloneqq\max(X_i, aX_j),\quad i,j \in V.
\end{equation}
The next theorem exploits the asymmetry between the scalings of $M_{i,j}$ and $M_{i,aj}$ in providing a criterion that helps to identify the source nodes in an LSEM+.

\begin{theorem}	\label{p2sourcenodes}
	Let $\boldsymbol{X}\in {\rm RV}_+^d(2)$ be an LSEM+ with innovation coefficient matrix $A$. Then $j$ is a source node on a well ordered DAG if and only if there exists a scalar $a>1$ such that 
	\begin{align}\label{p2critinit}
	\sigma^2_{M_{i,aj}}-\sigma_{M_{ij}}^2 =a^2-1, \quad \text{for all}\quad i\neq j.
	\end{align} 
	If $j$ is not a source node, then $\sigma^2_{M_{i,aj}}-\sigma_{M_{ij}}^2 \leq a^2-1$, with strict inequality if $i\in\an(j).$
	\end{theorem}

In order to illustrate the difference between this result and Theorem 2 in \cite{KK} we re-consider their Example 3 with $d=4$.

\begin{example}\label{p2theo1ex1}
Let $\boldsymbol{X}$ satisfy the setting in Theorem \ref{p2sourcenodes} with innovation coefficient matrix $A$ and corresponding {\DAG} given in Figure \ref{p2fig1}.
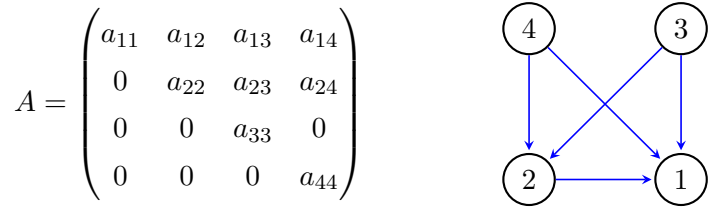
\begin{figure}[H]
	\centering
	\begin{tabular}{p{6cm}c}
		{$\displaystyle
			A={\renewcommand{\arraystretch}{1.2}
				\begin{pmatrix}
				a_{11}&a_{12} & a_{13}& a_{14} \tikzmark{lineone}\\
				0&a_{22}& a_{23}&a_{24}\tikzmark{linetwo}\\
				0&0&a_{33}&0\tikzmark{lineThree}\\
				0&0&0&a_{44}\tikzmark{lineFour}\\
				\end{pmatrix}}
			$}
		&$\vcenter{\hbox{\begin{tikzpicture}[
				> = stealth, 
				shorten > = 1pt, 
				auto,
				node distance = 2cm, 
				semithick 
				]
				\tikzstyle{every state}=[
				draw = black,
				thick,
				fill = white,
				minimum size = 4mm
				]
				\node[state] (4) {$4$};
				\node[state] (3) [right of=4] {$3$};
				\node[state] (2) [below of=4] {$2$};
				\node[state] (1) [below  of=3] {$1$};

				\path[->][blue] (4) edge node {} (2);
				\path[->][blue] (3) edge node {} (2);
				\path[->][blue] (3) edge node {} (1);
				\path[->][blue] (2) edge node {} (1);
				\path[->][blue] (4) edge node {} (1);
				\end{tikzpicture}}}$
	\end{tabular}
	\caption{Innovation coefficient matrix $A$ and corresponding DAG.}\label{p2fig1}
	\end{figure}
	
	We first investigate whether node $2$ is a source node for the DAG in Figure~\ref{p2fig1}. Set $a>1$, $i=4$ and $j=2$, and compute
	\begin{align}\label{p2exineq}
	\sigma_{M_{4,a2}}^2-\sigma_{M_{4,2}}^2=a^2(a_{22}^2+a_{23}^2)+(a^2a_{24}^2)\vee a_{44}^2-(a_{22}^2+a_{23}^2+a_{44}^2)< (a^2-1).
	\end{align}
The last inequality follows on noting that if $(a^2a_{24}^2)\vee a_{44}^2=a_{44}^2$, then $(a^2a_{24}^2)\vee a_{44}^2-a_{44}^2=0$; if not, then  $(a^2a_{24}^2)\vee a_{44}^2=a^2a_{24}^2$, in which case $(a^2a_{24}^2)\vee a_{44}^2-a_{44}^2<(a^2a_{24}^2)-a_{24}^2=(a^2-1)a_{24}^2$, by Lemma \ref{p2ineq}.
The last inequality, together with Assumption (A3), whereby the normalised row-norms give $a_{22}^2+a_{23}^2+a_{24}^2=1$, produces the upper bound $(a^2-1)(a_{22}^2+a_{23}^2+a_{24}^2)=a^2-1$ in \eqref{p2exineq}. Thus node $2$ does not satisfy Theorem~\ref{p2sourcenodes}, so it is not a source node.

If $j=4$, then for $i \in\{1,2,3\}$ Lemma~\ref{p2ineq} implies that
\begin{align*}
	\sigma_{M_{i,a4}}^2-\sigma_{M_{i,4}}^2=\sum_{j\neq 4}a_{ij}^2+a^2a_{44}-\sum_{j\neq 4}a_{ij}^2-a_{44}=a^2-1.
\end{align*} 
Hence Theorem~\ref{p2sourcenodes} entails that $4$ is a source node. The same equalities hold for $j=3$, so both $3$ and $4$ are source nodes. 

Theorem 2 of \citet{KK}, by contrast, only requires the computation of the difference $\sigma_{M_{1,2,3,a4}}^2-\sigma_{M_{1,2,3,4}}^2$, which is  derived from a  four-dimensional angular measure.
\end{example}

\subsection{Identification of descendants} \label{p2desc}

In this section we assume that we have identified and reordered $p$ nodes of the DAG and grouped them into an ordered set $I\coloneqq(i_{d-p+1},\ldots, i_d)$, so that $i_{k_1}>i_{k_2}$ if $i_{k_1}\in \pa(i_{k_2})$; we write $I^c=V\setminus I$ for the set of remaining unidentified (or unordered) nodes.  We assume that there are no ancestors of $I$ amongst the unidentified nodes, i.e., $\An(I)=I$. A typical example for $I$ is the set of source nodes, which clearly satisfies these properties and can be identified by employing Theorem~\ref{p2sourcenodes}, as demonstrated in Example~\ref{p2theo1ex1}.


Similarly to \eqref{p2Mmam}, but accounting for the non-empty ordered set $I$, we define for $a>1$,
\begin{equation}\label{p2Mmam2}
M_{i,aj, aI} \coloneqq\max(X_i, aX_j, a\bs X_I), \quad i, j \in I^c.
\end{equation}

The following theorem exploits asymmetries in the scalings of $M_{i,aj, aI}$ and provides a criterion for node $j$ to have no ancestors other than the identified ones in $I$.

\begin{theorem}\label{p2descnodes}
	Let $\boldsymbol{X}\in {\rm RV}_+^d(2)$ be an LSEM+ with innovation coefficient matrix $A$. Let $I$ denote the set of identified and ordered $p$ nodes having no ancestors outside~$I$. 
	Then $j\in I^c$ has no ancestors outside~$I$ 
	if and only if there exists a scalar $a>1$ such that
	\begin{equation}
	\sigma_{M_{i, aj, aI}}^2-\sigma_{M_{i, j, I}}^2=(a^2-1)\sigma_{M_{j, I}}^2,\quad \text{for all}\quad  i\notin I\cup\{j\}.
	\label{p2critpair}
	\end{equation}
	Otherwise, $\sigma^2_{M_{i,aj,aI}}-\sigma_{M_{ijI}}^2 \leq (a^2-1)\sigma_{M_{j, I}}^2$, with strict inequality for  $i\in I^c\cap\an(j)$.
\end{theorem}

When two nodes, say $j_1$ and $j_2$, satisfy eq. (\ref{p2critpair}), then there is no causal link between them, as summarised in the next corollary.

\begin{corollary}\label{p2multnode}
	In the setting of Theorem \ref{p2descnodes}, suppose that the nodes $j_1, j_2\notin I$ satisfy (\ref{p2critpair}). Then $j_1\notin\an(j_2)$ and $j_2\notin\an(j_1)$.
\end{corollary}

We illustrate Theorem \ref{p2descnodes} by reconsidering the DAG in Example \ref{p2theo1ex1}.

\begin{example}\label{p2theo2ex2}
	Suppose that in Example \ref{p2theo1ex1} we have identified the two source nodes $3$ and $4$: hence $I=(3,4)$. Let $j=1$, and suppose we want to check whether $1\in \des(2)$. Then,
		\begin{align*}
		\sigma_{M_{2, a1, aI}}^2-\sigma_{M_{2, 1, I}}^2&=
	a^2a_{11}^2+\left(a^2a_{12}^2\right)\vee a_{22}^2+a^2a_{33}^2+a^2a_{44}^2- \sum_{k=1}^{4}a_{kk}^2\\
		&=(a^2-1)\sum_{k\neq 2}a_{kk}^2+\left(a^2a_{12}^2\right)\vee a_{22}^2-a_{22}^2\\
		&<(a^2-1)\sum_{k\neq2}a_{kk}^2+(a^2-1)a_{12}^2=	(a^2-1)\sigma_{M_{I, j}}^2,
		\end{align*}
		where the first equality applies Lemma \ref{p2scalcoll}, and the inequality follows from Lemma \ref{p2ineq} and arguments similar to those in \eqref{p2exineq}.  Since the inequality is strict, Theorem \ref{p2descnodes} implies that $1\in\des(2)$. Instead, if we take $i=1$ and $j=2$ and follow the same procedure, we reach equality (\ref{p2critpair}), indicating that $2\notin\des(1)$. This gives the causal ordering $I=(1, 2, 3, 4)$.
\end{example}

\subsection{A consistent algorithm to estimate a causal ordering}\label{p2consist}
We now combine the results of Sections~\ref{p2sources} and \ref{p2desc} into an algorithm to identify a causal ordering.
We use Theorem \ref{p2sourcenodes} to initialise Algorithm \ref{p2causordalg} (with $I=\emptyset$), and then apply Theorem \ref{p2descnodes} until the $d$ nodes have been ordered. We defer details on the estimation of the scalings to Appendix \ref{p2estprocedure}, where we also show that their consistency translates into consistency of our algorithm. We first outline the elements of Algorithm \ref{p2causordalg}:
\begin{itemize}
	\item[-]  the matrix $\Delta_{\hat I} \in \mathbb{R}^{d\times d}$, with entries \begin{align}\label{p2algupdt}
	(\Delta_{\hat I})_{ij}&=\hat{\sigma}^2_{M_{i, a j, aI}}-\hat{\sigma}^2_{M_{i, j, I}}-(a^2-1) \hat{\sigma}^2_{M_{ j, I}},\quad  i,j\in V \setminus I;\\
	(\Delta_{\hat I})_{ij}&=\infty,\quad  i{\rm \,or\,}j\in I;\nonumber
	\end{align}
	\item[-]  the term $ \eps_{\hat I}\coloneqq \eps\left| {\rm max}\left({\rm colMins}_{I^c} \left(\Delta_{\hat I}\right)\right)\right|$, where $\eps>0$ and the operator colMins takes the entrywise minimum of each column not indexed in $I$ of the matrix $\hat{\Delta}_I$; and
	\item[-] the vector ${\bs \delta}_{\hat I}=({\delta}_{\hat I,1},\ldots, {\delta}_{\hat I,d} )$.
\end{itemize}

At each iteration of the {\bf while} loop in Algorithm \ref{p2causordalg},  we update ${\Delta}_{\hat I}$ and $\eps_{\hat I}$ by accounting for the set $I$ of identified nodes.

We briefly illustrate the motivation behind Algorithm \ref{p2causordalg}, in particular for the $\eps$ term. To do so, we reconsider Example \ref{p2theo1ex1} and go through the first iteration when $I=\emptyset$.

\begin{algorithm}[t]
	\caption{Estimation of a causal order of the LSEM+ $\boldsymbol{X}$}\label{p2causordalg}
	\textbf{Input}: $(\bs X_i: i\in\{1,\ldots,n\})$  from $\bs X\in{\rm RV}^d_+(2)$ with standard margins, $a>1, \eps>0,$ \hspace*{1.5cm}$I=\emptyset, \Delta_{\hat I} = (0)_{d\times d}, {\bs \delta}_{\hat I} = (0)_{1\times d}; $\\
	\textbf{Output}: The ordered set $I$
	\begin{algorithmic}[1]
		\Procedure{}{} 
		\State \textbf{while} $| I|< d$ \textbf{do}
		\State \hspace{5mm} \textbf{Update}  $\Delta_{\hat I} $ \textbf{using }\eqref{p2algupdt}
		%
		\State \hspace{5mm} \textbf{Set} ${\bs \delta}_{\hat I}= { {\rm colMins}_V (\Delta_{\hat I})} - {\rm max(colMins}_{I^c} (\Delta_{\hat I}))$
		\State \hspace{13.4mm} $\pi = \underset{\{\,p\,\in\, I^c\,:\, |\delta_{\hat I,p}|\, \leq\, \eps_{ \hat I} \}}{\rm arg\, \rm sort} \bs \delta_{\hat I}$
		\State \hspace{5mm} \textbf{Update $I$ by adding $\pi$}: $I\leftarrow (\pi,   I)$ 
		\State\textbf{end while} 
		\State\textbf{return} $I$ 
		\EndProcedure
	\end{algorithmic}
\end{algorithm}


\begin{remark}
Consider the LSEM+ $\bs X$ with DAG in Figure~\ref{p2fig1} and suppose that $I=\emptyset$. We look at the theoretical variants of the estimates used in Algorithm~\ref{p2causordalg}, i.e., the $d\times d$ matrix $\Delta_I$, computed from the theoretical scalings.  

For $j\in\{3,4\}$ and $i\neq j$, it follows from Example~\ref{p2theo1ex1} that the vector {\rm colMins}$_V(\Delta_I)$ equals $(c_1, c_2, c_3, c_4)$, with $c_3=c_4=0$, $c_1$ and $c_2$ strictly negative, and $\eps_I=0$. Line 4 of the algorithm gives $\bs \delta_I=(c_1, c_2, 0,0)$. For finitely many observations, however, the estimated vector $\bs \delta_{\hat I}=(\hat{c}_1,\hat{c}_2,\hat{c}_3,\hat{c}_4)$ is almost surely different from $\bs \delta_I$, and,  if $\eps>0$, then $\eps_{\hat I}>0$. If $\eps=0$ we may then only select one of the source nodes between $3$ and $4$, but if we let $\epsilon>0$, we allow for a  difference between $\hat{c}_3$ and $\hat{c}_4$, whereby we may select nodes 3 and 4 as source nodes in line 5 of the algorithm.  The error term preserves consistency in selecting the correct nodes: to see why, note that $\hat{c}_3$, $\hat{c}_4$ and $\eps_{\hat I}$  approach zero in probability, and $\hat{c}_1$ and $\hat{c}_2$ respectively converge to $c_1$ and $c_2$. The introduction of the $\eps$ term therefore enables the  selection of more than one source node at a time. When $I\neq\emptyset$, a similar reasoning applies to the nodes with ancestors in $I$.
\end{remark}

The methodology proposed in \citet[\S~7]{KK} relies on fixed error terms $\eps_1$ and $\eps_2$ which, if too small, may prevent their algorithm from re-ordering the nodes.
In contrast, Algorithm~\ref{p2causordalg} updates $\eps_{\hat I}$ for each iteration, ensuring that a causal ordering is~returned.

The next proposition, proved in Appendix~\ref{p2consistencyy}, establishes consistency of Algorithm \ref{p2causordalg}.

\begin{proposition}\label{p2consistencyy}
	Let $\boldsymbol{X}\in {\rm RV}_+^d(2)$ be an LSEM+ with coefficient matrix $A$. Let $I=(i_1,\ldots, i_d)$ denote the output of Algorithm \ref{p2causordalg}. Then $I$ is a consistent estimator of a causal ordering of the DAG supporting $\boldsymbol{X}$.
\end{proposition}

\section{Simulation study and data application \protect\footnote{A statistical package containing code for the proposed methodology is under development and will soon be available on \texttt{github.com/mariokrali}\,.}} \label{p2simdata}

\subsection{Methods}
In this section we evaluate the empirical performance of Algorithm \ref{p2causordalg} by simulation 
and an application to two river discharge datasets. We compare  our method to EASE \citep{gnecco}, the only alternative method that can infer causal orderings from extreme observations for the LSEM. EASE is  based on the causal tail coefficient,
$$\Gamma_{ij}=\lim\limits_{u\to 1^-}\mathbb{E}[ F_j(X_j)|  F_i(X_i)>u ],
$$
where $F_i$ and  $F_j$ denote the distribution functions of $X_i$ and  $X_j$. The $\Gamma_{ij}$ are computed for all available pairs to estimate a causal ordering. To implement EASE we use the R package \texttt{causalXtremes} \citep{gnecco}. 

\subsection{Simulation setup}\label{p2simsetup}
The simulations are carried over random DAGs which are used to support the LSEM+'s. To give a variety of such models, we consider several configurations of the dimension $d$ and the sparsity $p$ of the DAGs, and of the parameters of the LSEM+, namely the index of regular variation $\alpha$ and the sample size $n$. For every combination of $d\in\{20, 40, 50\}$, $p\in\{0.05, 0.1\}$, $\alpha\in\{2,3\}$ and $n\in\{1000, 5000\}$, we generate 50 random DAGs and compare the two methods.

We construct each DAG via its upper-triangular adjacency matrix $C\in\mathbb{R}_+^{d\times d}$, with edges drawn independently from a Bernoulli distribution with success probability $p$. The edge weights in $C$ and the innovation weights $s_{ii}$ of $S$ in \eqref{p2lsem} are drawn as independent $\mathcal{U}[0.1, 1.5]$ random variables. 

The innovation vector $\boldsymbol{Z}$ has independent $|t_\alpha|$-distributed components, where $t_\alpha$ denotes a Student-$t$ distribution with $\alpha$ degrees of freedom.

As choice of input parameters for Algorithm~\ref{p2causordalg} we consider $\epsilon \in\{0.1, 0.4\}$ and fix the scalar parameter $a=1.3$. Appendix \ref{p2ascal} investigates its stability for other choices of $a$.

To estimate the scaling parameters for Algorithm~\ref{p2causordalg}, we employ the empirical probability integral transform to first standardise the original observations to standard Fr\'echet(2) margins, and then use the $k$ largest thresholded observations. This procedure is described in Appendix~\ref{p2estprocedure}.
We apply the one-tailed version of EASE with the prespecified choice $k=n^{\lfloor 0.4\rfloor}$, which was found to yield the best performance in \cite{gnecco}.

\subsection{Evaluation metrics}
To evaluate the performance of Algorithm \ref{p2causordalg} we use the Structural Intervention Distance (SID) metric \citep{sid} which was used to assess the performance of EASE in \cite{gnecco}. This  metric focuses on causal relations in a DAG and employs intervention distributions via parental adjustments between pairs of nodes. This renders SID the natural choice in evaluating differences between causal orderings. The metric can be used to quantify the distance between two DAGs, or a DAG and a completed partially directed acyclic graph (CPDAG). We use the implementation of the metric in the package \texttt{causalXtremes} \citep{gnecco}, where SID is normalised to take values in $[0,1]$, with lower values indicating a smaller distance and thus a better performance. A SID equal to zero indicates valid causal orderings on both DAGs. In the simulation study we compute the SID between the true DAG and the fully connected DAG corresponding to the estimated causal ordering. The latter is obtained by adding the edges $j\to i$ if $j>i$.

\subsection{Simulation results}\label{p2simres}
Figures~\ref{p2runs1000alt} and \ref{p2vsruns1000} summarise our findings. The performances of both algorithms are affected by the choice of the parameters. In particular, both Algorithm \ref{p2causordalg} and EASE perform worse in higher-dimensional and denser regimes, which give rise to more complex dependence structures with more edges and parental adjustments. The difference in performance  can be seen by comparing SID between the settings with $p=0.05$ and $p=0.1$ in Figure~\ref{p2runs1000alt}.

Figure~\ref{p2vsruns1000} shows that Algorithm~\ref{p2causordalg} performs better in the sparser regime for a suitable choice of threshold $k$, whereas EASE is better in the denser, higher-dimensional setup and for the heavier tail, i.e., $(d, p, \alpha)=(50, 0.1, 2)$.

\begin{figure}[htbp]\centering
	\hspace{-0.4cm}\includegraphics[ height=5cm, width=13cm, clip]{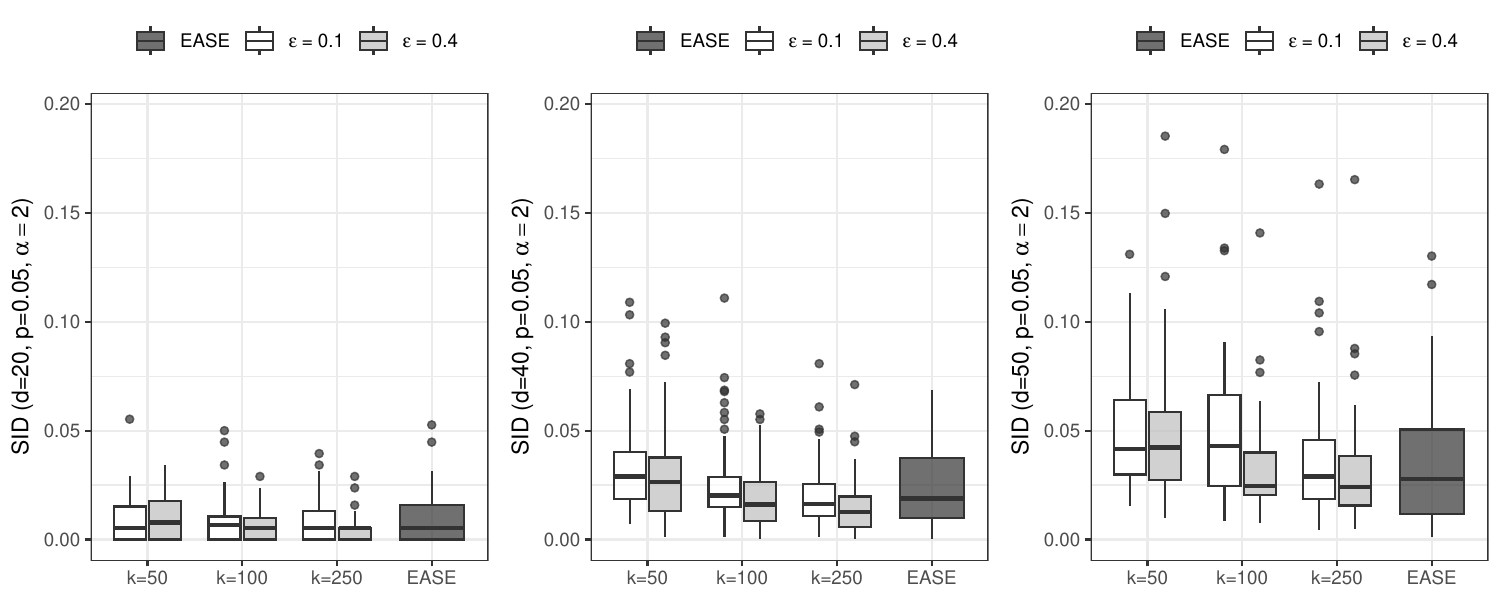}
	
	\hspace{-0.4cm}\includegraphics[ height=4.3cm, width=13cm, clip]{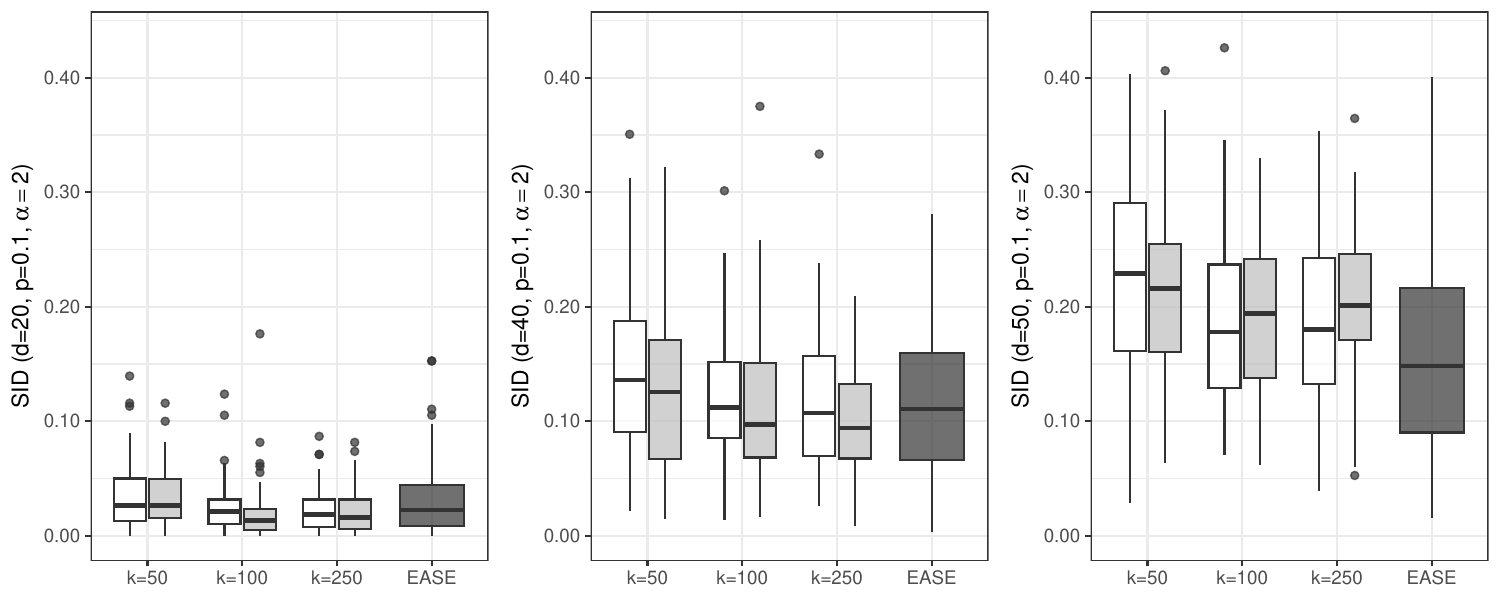}
	
	\hspace{-0.4cm}\includegraphics[ height=4.3cm, width=13cm, clip]{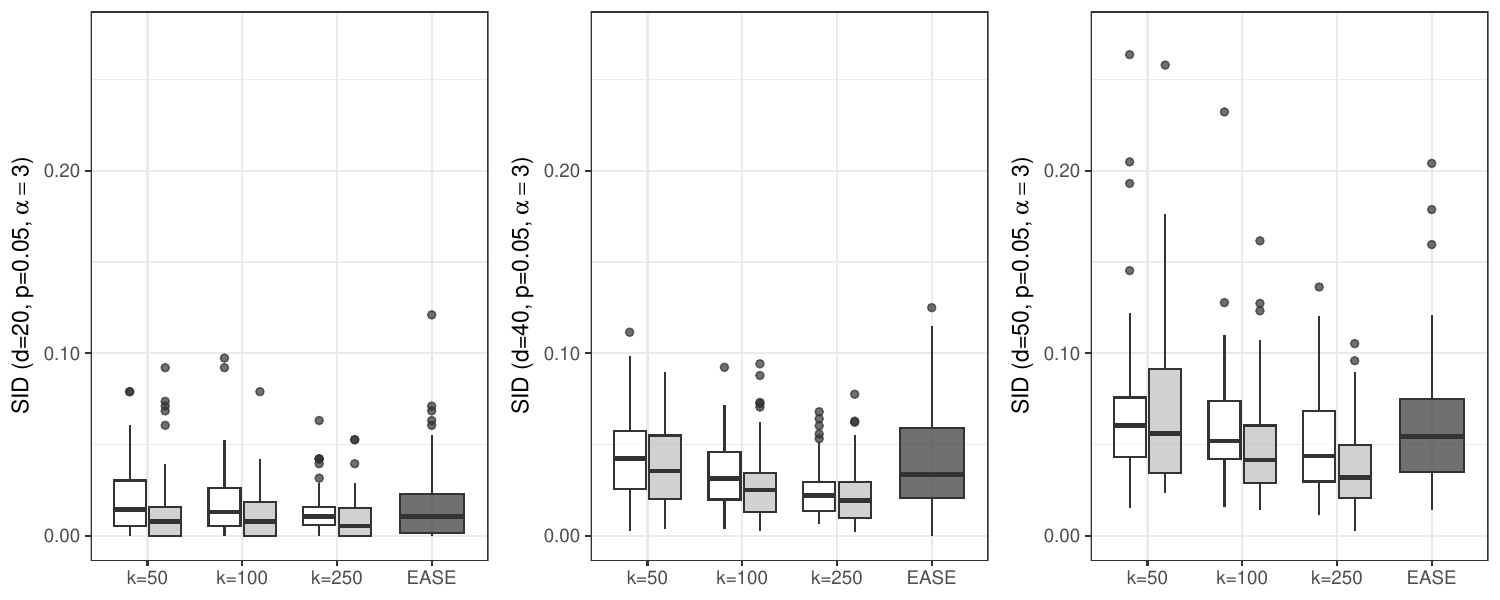}
	
	\hspace{-0.4cm}\includegraphics[ height=4.3cm, width=13cm, clip]{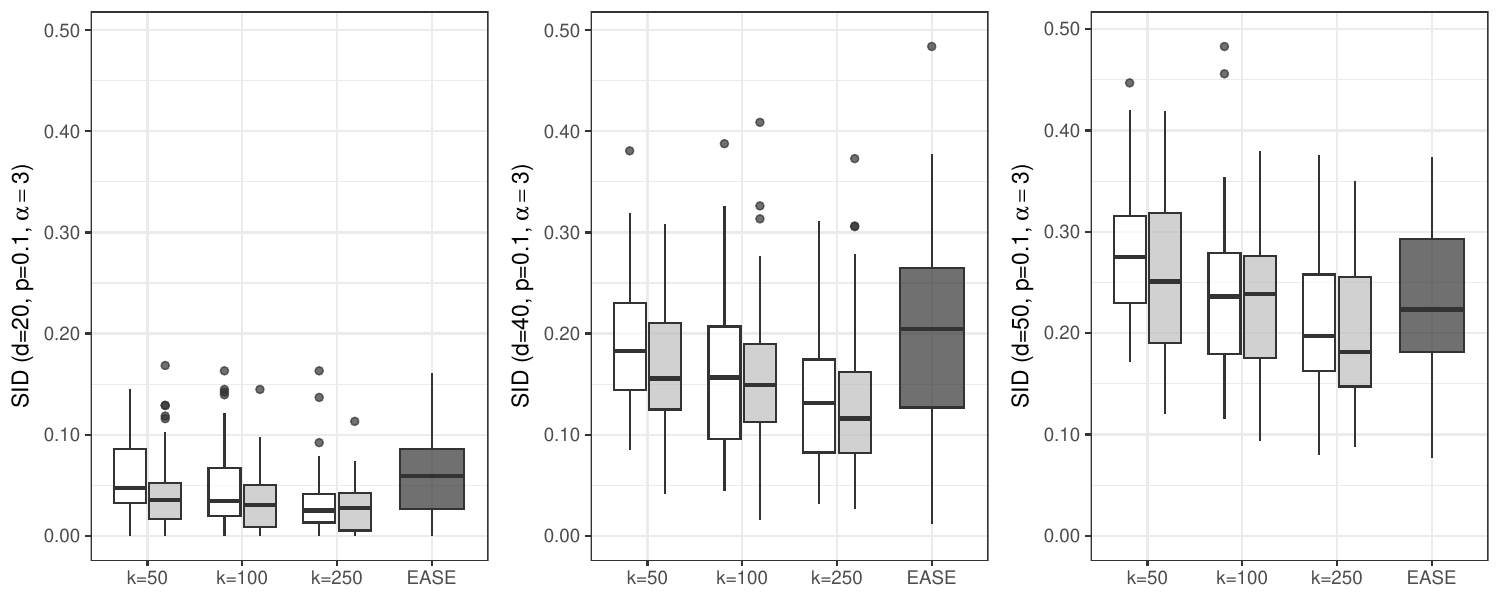}
	
	\caption{Boxplots of the SID of the causal orderings of Algorithm \ref{p2causordalg} and of EASE evaluated for 50 random DAGs for each configuration of $(d, p, \alpha)$ and for $n=1000$.} \label{p2runs1000alt}
\end{figure}

\begin{figure}[htbp]\centering
	\hspace{-0cm}\includegraphics[ height=4cm, width=12cm, clip]{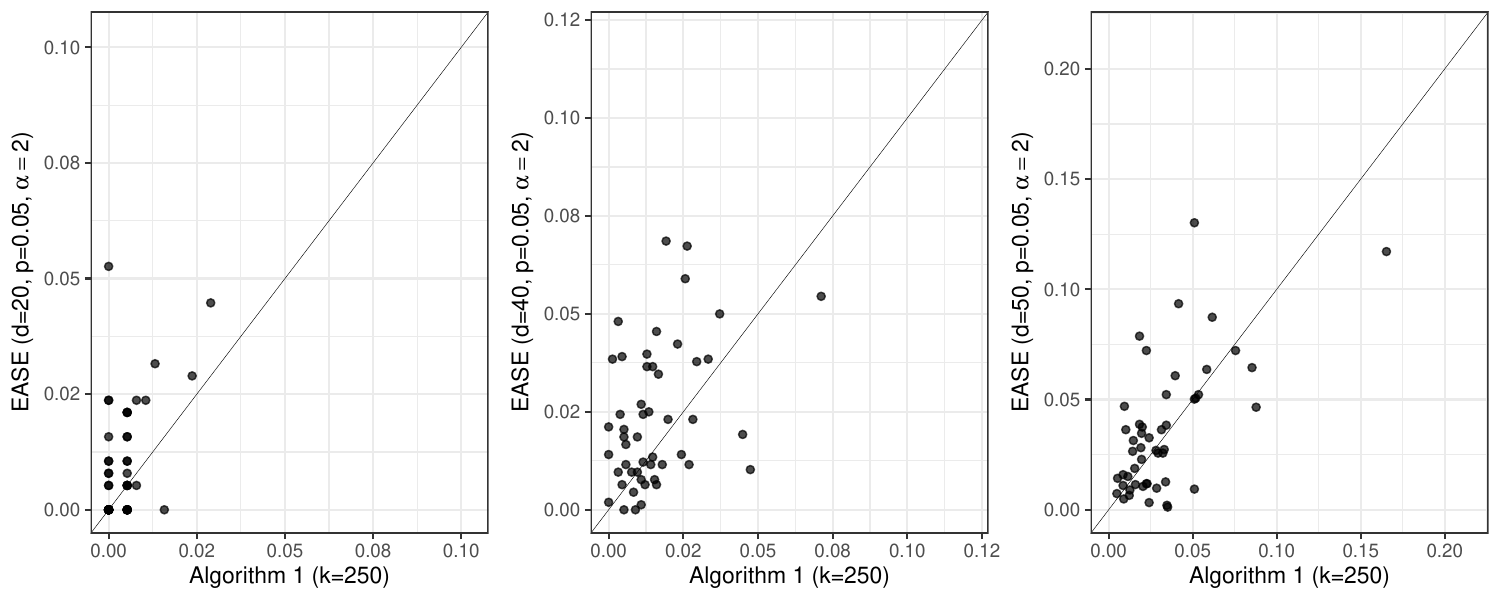}
	
	\hspace{-0cm}\includegraphics[ height=4cm, width=12cm, clip]{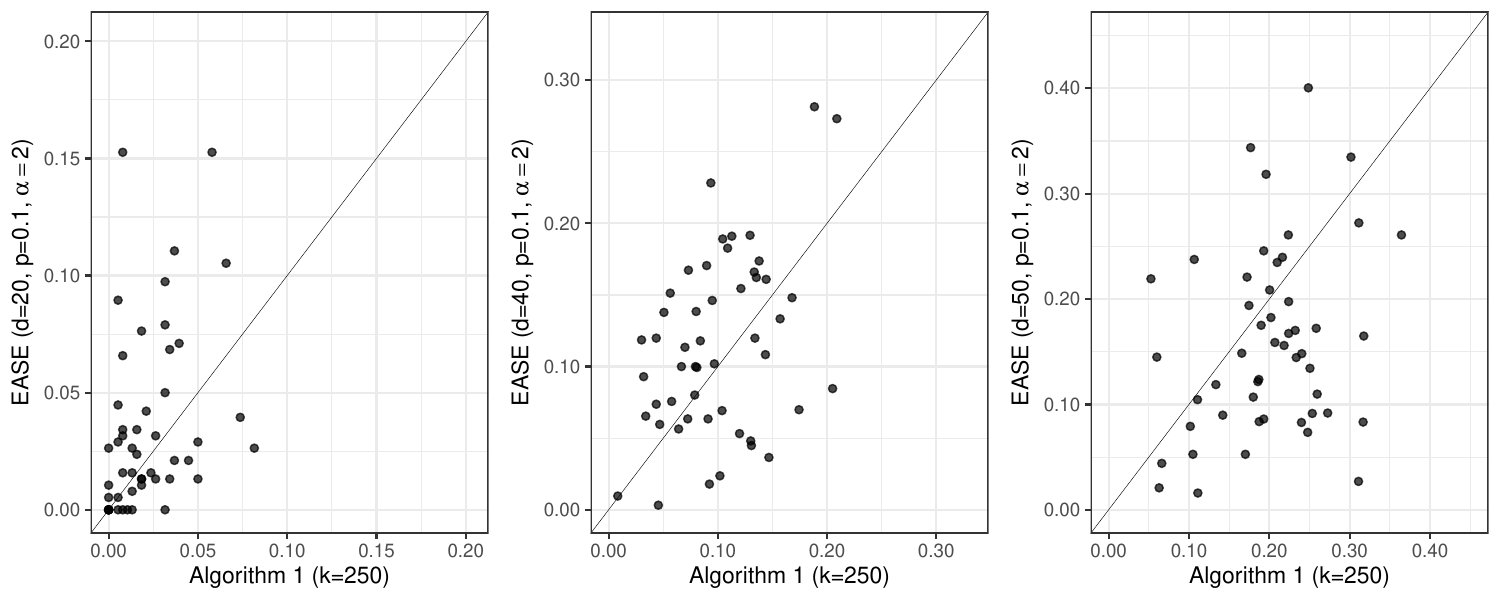}
	
	\hspace{-0cm}\includegraphics[ height=4cm, width=12cm, clip]{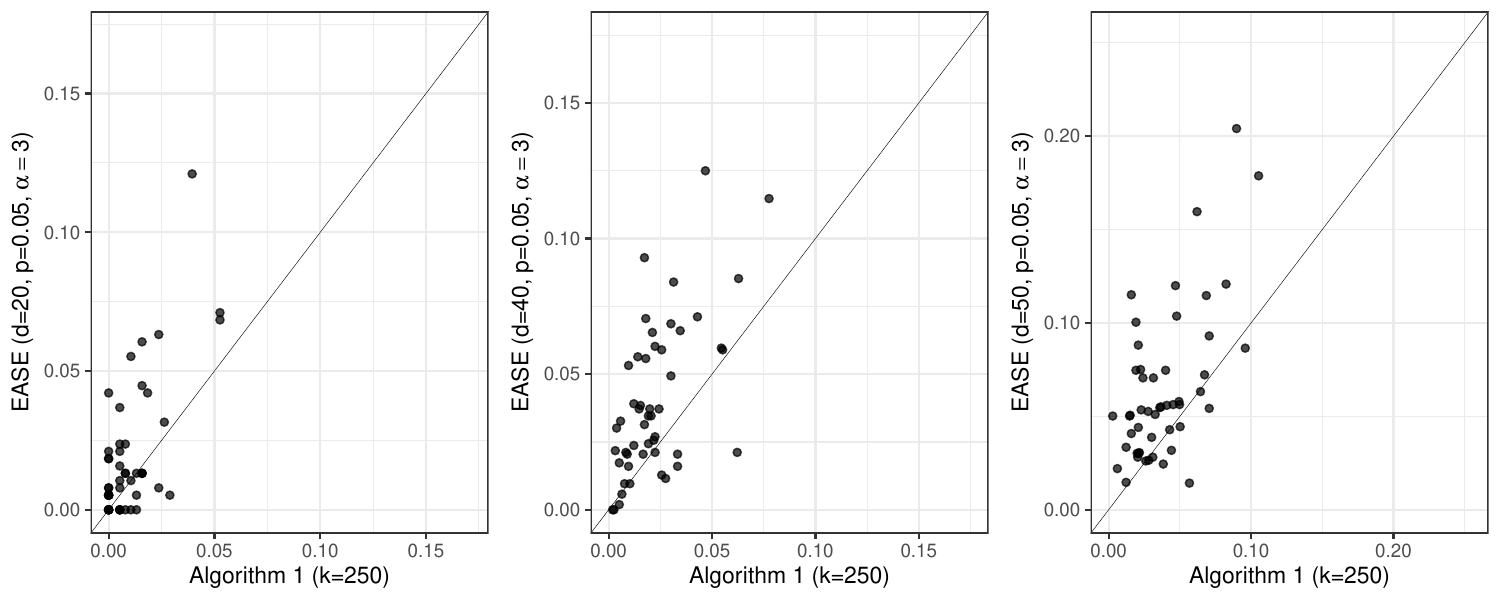}
	
	\hspace{-0cm}\includegraphics[ height=4cm, width=12cm, clip]{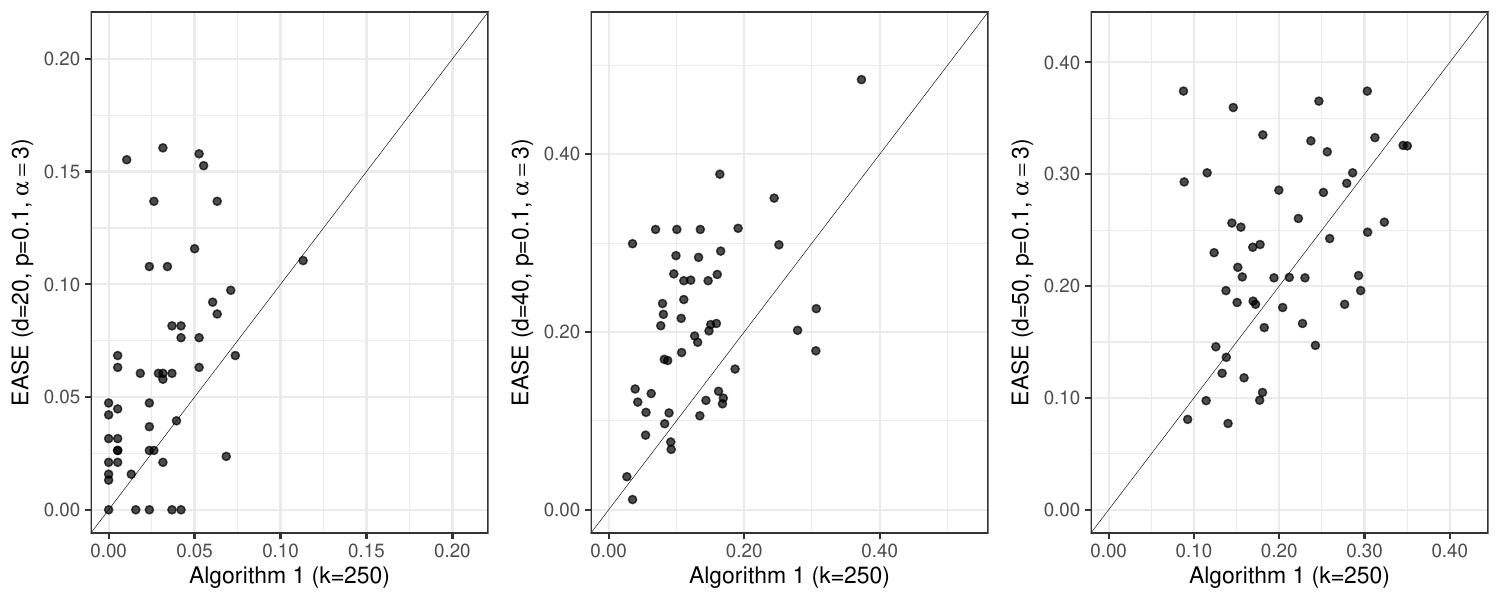}
	
	\caption{Scatterplots of the SID of the causal orderings of Algorithm~\ref{p2causordalg} and of EASE evaluated for 50 random DAGs for each configuration of $(d, p, \alpha)$ and for $n=1000$. }\label{p2vsruns1000}
\end{figure} 

The performances of the algorithms are affected by the regular variation index $\alpha$, which influences the rate of convergence of the respective componentwise maxima to their limiting Fr\'echet distributions \citep[Prop.~2.12]{sres}. In particular, a lower value of $\alpha$ corresponds to a quicker convergence rate and leads to better performance of both methods. 
However, a comparison between the configurations $(d,p,\alpha)=(50,0.1, 2)$ and $(d,p,\alpha)=(50,0.1, 3)$ shows a larger difference in SID for EASE, which is less robust to lighter tails.

Finally, the results for $n=5000$ are provided in Appendix \ref{p2bplot5000} and show better performance in both methods,  with Algorithm~\ref{p2causordalg} still outperforming EASE in the sparser configuration.
We also see greater stability of the former for different thresholds  $k$. EASE retains a mild advantage over  Algorithm \ref{p2causordalg} in the denser setup and for $\alpha=2$, perhaps because of the  faster convergence of the tail dependence with higher $n$, and the fact that EASE accounts for hidden confounders in pairwise dependencies. Algorithm~\ref{p2causordalg} performs better for $\alpha=3$, even in the denser regime.

As our method estimates the scalings from the angular measure of a vector of dimension up to $d$, as described in Appendix \ref{p2estprocedure}, it also requires more thresholded observations $k$ to attain a better performance. This can be seen for the larger values of $k$ when $n=1000$. 

\subsection{Data application}
We now attempt to infer causal orderings for datasets of river discharges from the upper Danube and the Rhine basin in Switzerland. 
Given the results in the simulation study and  the sparse regime of the networks, we fix $a=1.3$ and $\eps=0.4$ in Algorithm~\ref{p2causordalg}. The procedure for estimating the scalings is identical to that in the simulation study.

The DAG adjacency matrices of the networks are constructed based on the physical flow connections between the gauging stations. 

To evaluate the robustness of the algorithms when the assumption of independent and identically distributed observations fails, we study the stability of the SID metric by analysing observations from both unprocessed and declustered datasets, following a multivariate declustering procedure introduced by \cite{engelke2015} and described in Section \ref{p2declustprocess}. This leaves only those observations that are approximately independent in time. To assess uncertainty in the causal orderings, we evaluate the SID metric over 100 bootstrap replications.

\subsection{Upper Danube basin}
The Danube dataset, available from \url{http://gdk.bayern.de}, or as a supplemental file of \cite{engelke2015}, has become a benchmark for testing methodologies in extremal graphical modelling \citep{engelke:hitz:18, gnecco, mhalla, TBK}. The data span the months June--August from 1960 to 2010, thus eliminating any seasonality, and consists of $n=4692$ daily discharges at $d=31$ gauging stations, depicted in Figure~\ref{p2danubedag}.

\begin{figure}[htbp]\centering
\vspace{-.25cm}	\includegraphics[ height=7cm, width=7cm, clip]{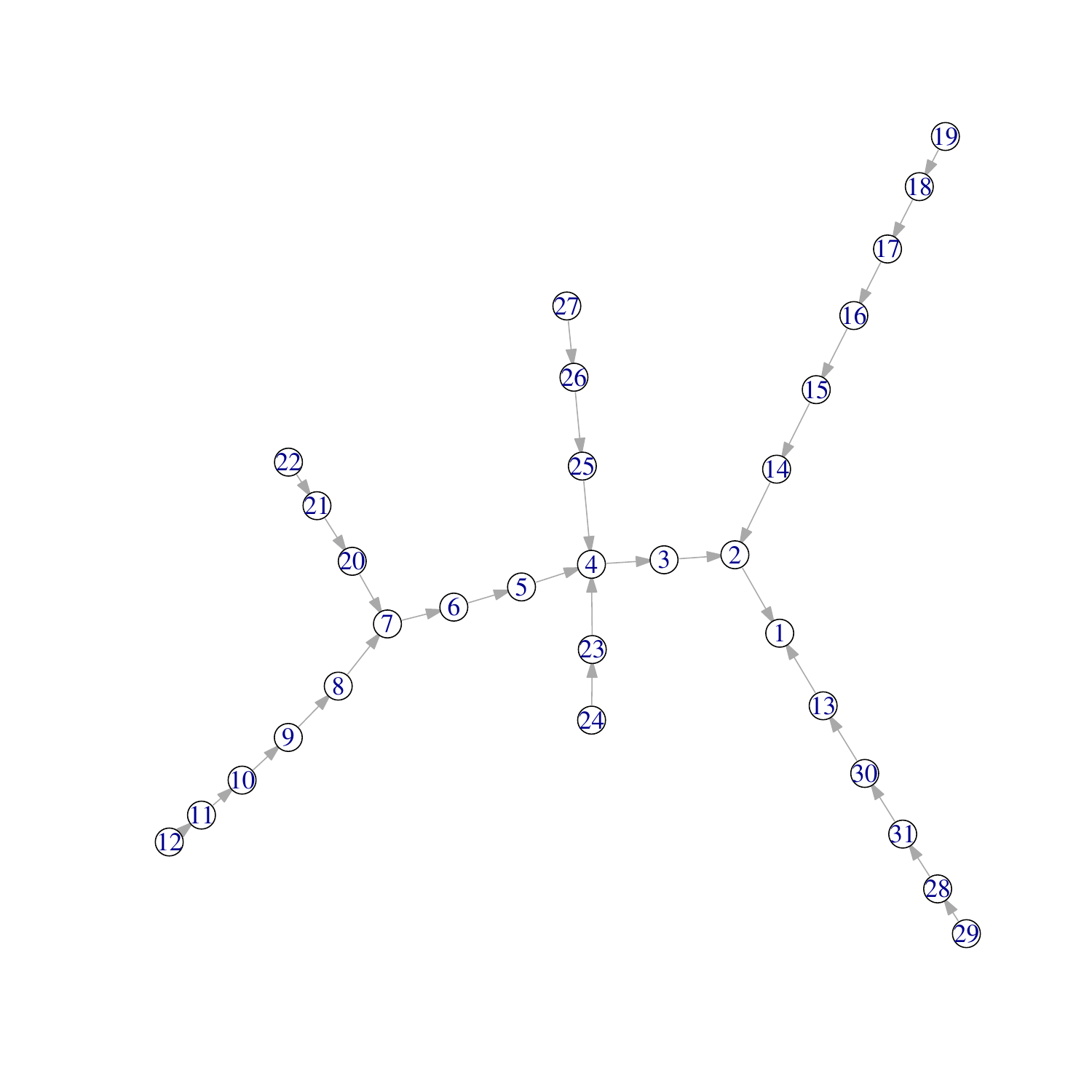}\vspace{-1cm}
	\caption{Upper Danube basin network induced by physical flow connections \citep{engelke2015}.} \label{p2danubedag}
\end{figure}

\subsubsection{Unprocessed Danube data}
The analysis of \cite{engelke2015} indicates large variations in discharges across the stations at different altitudes. Since our methodology requires ${\rm RV}_+(2)$ random variables, we employ the empirical probability integral transform  to obtain standard Fr\'echet(2) margins (see  Appendix~\ref{p2estprocedure}). 
Figure~\ref{p2danubebootst} summarises the performances of the methods.

\subsubsection{Declustered Danube data}\label{p2declustprocess}
Both Algorithm~\ref{p2causordalg} and EASE yield consistent causal orderings for independent observations from an LSEM+. Independence, however, is questionable, as high discharges typically persist for a number of days, and thus cluster in time. To account for this we apply the multivariate declustering procedure of \citet{engelke2015}, which sequentially selects non-overlapping windows  of $l$ days of observations for each of the 50 summer periods. These windows are initially taken around the highest observations in each series and across all stations. Each window retains only the largest observation; we repeat this procedure until there are no time-windows with $l$ consecutive observations.
This processing step with time windows of width $l=9$ yields $n=428$ approximately independent observations.

Figure~\ref{p2danubebootst} provides the SID for the declustered Danube data. We assess its uncertainty for 100 bootstrap replicates, which show that, despite the dependence in time, Algorithm~\ref{p2causordalg} performs better than EASE for both the unprocessed and the declustered data. 

\begin{figure}[t]\centering
\includegraphics[ height=5cm, width=13cm, clip]{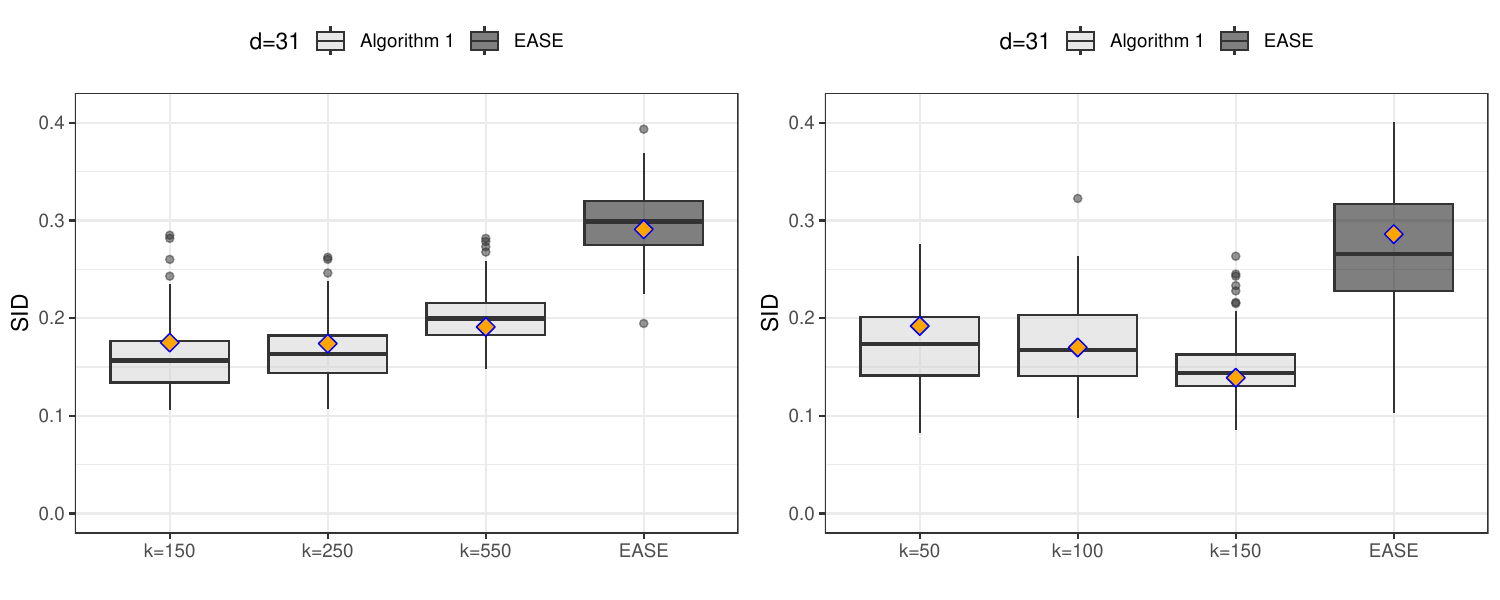}
	
	\caption{Boxplots of the SID of the causal orderings of Algorithm~\ref{p2causordalg} and of EASE 
	for 100 bootstrap replicates of the unprocessed (left) and declustered (right) Danube datasets. The orange diamonds show the SID evaluated for the observed data for each method.}\vspace{-2mm} \label{p2danubebootst}
\end{figure}

\subsection{Rhine basin}
This dataset consists of daily river discharge amounts from $d=68$ gauging stations along the Rhine basin in Switzerland and was studied in \cite{asadiregio}. The data range from 1913 to 2014, but we only select those  $n=2024$ observations when measurements are available across all 68 stations.

To construct the adjacency matrix of the DAG underlying the graphical model of the river network, we start from the physical flow connections in \cite{asadiregio}, and put edges only between neighbouring stations, with direction corresponding to downstream flow. Station $67$ has the highest in-degree due to incoming flows from most of the stations; see Figure \ref{p2rhinedag}.

\begin{figure}[htpb]\centering
\vspace{-.7cm}	\includegraphics[ height=8cm, width=8cm, clip]{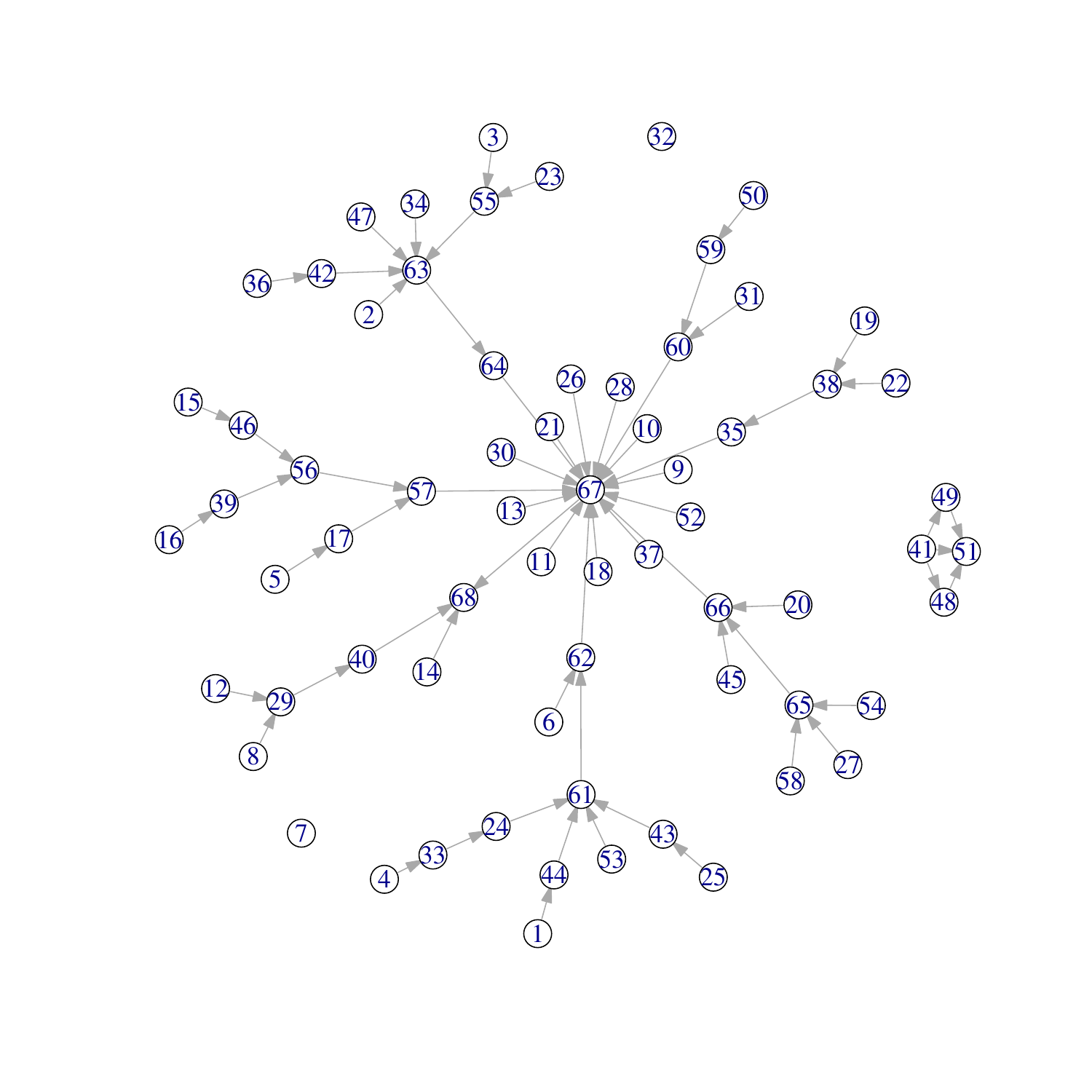}
\vspace{-1cm}	\caption{Rhine basin network induced by physical flow connections \citep{asadiregio}.} \label{p2rhinedag}
\end{figure}

Preliminary assessment of the pairwise dependencies amongst the stations reveals correlations above $0.99$ for three pairs of stations in close proximity, namely (48, 49), (63, 64) and (67,~68). As our methodology exploits asymmetries in scalings, such strong dependence complicates the identification of a causal direction, especially when coupled with high dimensionality and time dependence.

To see how robust Algorithm~\ref{p2causordalg} is to the presence of strong dependence amongst pairs in the river network, we conduct a parallel analysis by removing stations 48, 63 and 67. The new DAG is obtained by putting edges only between the $65$ neighbouring stations. For instance, for $k\in\pa(j)$ and $j\in \pa(i)$, we add edges $k\to i$ after removing node $j\in\{48, 63, 67\}$.
Since none of the removed nodes is a confounder, the reduced DAG has no hidden confounding.

\subsubsection{Unprocessed Rhine data}
In this section we work with the unprocessed discharges. We evaluate the performance of Algorithm \ref{p2causordalg} and EASE for both the full river network consisting of $d=68$ stations and the reduced one of $d=65$ stations.

The left-hand panel in Figure~\ref{p2fig:rhine_declust} indicates superior performance for EASE when analysing the full river network. The uncertainty boxplots for Algorithm~\ref{p2causordalg} show a pattern similar to the Danube data but with higher variation of  SID, possibly due to the higher dimensionality.

The effect of time dependence of high discharges is likely to be more pronounced in Algorithm \ref{p2causordalg}, which relies on a larger number $k$ of thresholded observations to approximate higher-dimensional angular measures.

The results with $d=65$ stations in 
Figure~\ref{p2fig:rhine_declust} show a similar performance from Algorithm~\ref{p2causordalg}, yet reveal a drastic worsening for EASE. This is surprising given that the latter also works  in settings with hidden nodes. Nevertheless, EASE relies on comparisons between pairwise causal tail coefficients, and leaving out nodes seems to make their algorithm prone to high variability, thus pointing to high sensitivity to the coefficients $\Gamma_{ij}$.

\subsubsection{Declustered Rhine data}
  In this analysis we decluster the observations using time windows of width $l=9$. The bottom panel in Figure~\ref{p2fig:rhine_declust} contains the results for the full $(d=68)$ and the reduced $(d=65)$ networks, and show an improved performance of Algorithm~\ref{p2causordalg} relative to unprocessed data, reflected in a reduction of the SID and of its uncertainty. For $d=68$ and $k=60$, Algorithm \ref{p2causordalg} performs like EASE, and its performance remains robust across both networks. Interestingly, we see the same pattern as in the top panel: the performance of EASE is adversely affected in the reduced network. Although declustering seemingly leads to an improvement of SID for the latter, it also increases its variability.  

\begin{figure}[t]\centering
\includegraphics[ height=5cm, width=13cm, clip]{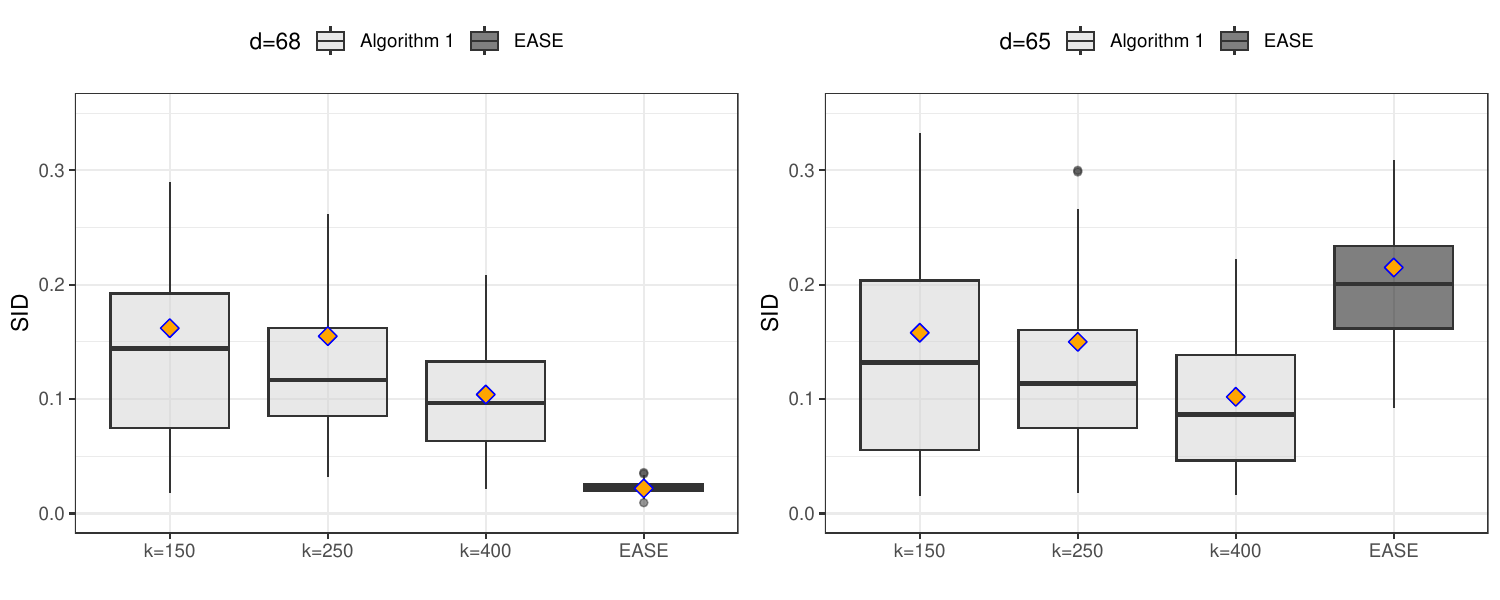}
\includegraphics[ height=5cm, width=13cm, clip]{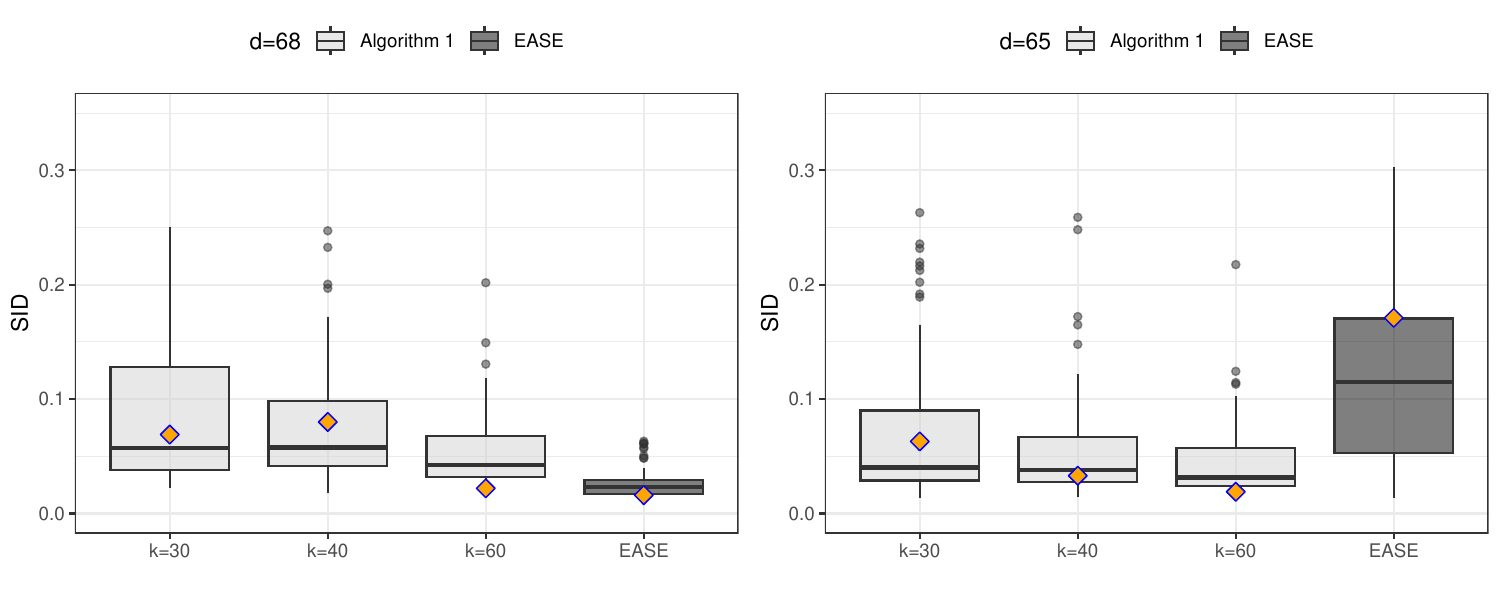}
	\caption{Boxplots of the SID of the causal orderings of Algorithm~\ref{p2causordalg} and of EASE
		 for 100 bootstrap replicates for the unprocessed (top panel) and declustered (bottom panel) Rhine datasets with $d=68$ (left-hand panel) and $d=65$ (right-hand panel) stations. The orange diamonds show the SID evaluated for the observed data for each method.}\vspace{-2mm}
\label{p2fig:rhine_declust}
\end{figure}


\section{Conclusion}

In this work we propose a scaling methodology for causal discovery in heavy-tailed linear structural equation models, based on scaling parameters derived from the angular measure that characterises multivariate extremal dependence. Following consistency of the estimators of the scalings we show that our causal discovery algorithm yields a consistent causal ordering. We employ the SID metric \citep{sid}, derived from structural interventions, to evaluate the performance of our algorithm in a simulation study and a data application to river discharges.  Comparison with EASE \citep{gnecco} shows that our methodology is competitive with the latter and can be more robust. 

Causal dependence modelling in extremes is in its infancy, and high-dimensional multivariate extremal analysis remains challenging. The data applications pinpoint vulnerabilities of both methods, with EASE being highly sensitive to pairwise dependencies of the node variables and to lighter tails, and Algorithm~\ref{p2causordalg} being sensitive to the increasing dimension of the angular measure in denser graphs. Interesting directions for future work include investigating the possibilities for further reduction in the dimension of the angular measure, which may aid in constructing better estimators for the scalings.

\subsection*{Acknowledgments}

I thank Anthony Davison for comments and suggestions that have improved the manuscript. I am grateful to the EPFL Doctoral School of Mathematics for the financial support.

\bibliography{texrefs}

\begin{thebibliography}{39}
\providecommand{\natexlab}[1]{#1}
\providecommand{\url}[1]{\texttt{#1}}
\expandafter\ifx\csname urlstyle\endcsname\relax
  \providecommand{\doi}[1]{doi: #1}\else
  \providecommand{\doi}{doi: \begingroup \urlstyle{rm}\Url}\fi

\bibitem[Asadi et~al.(2015)Asadi, Davison, and Engelke]{engelke2015}
P.~Asadi, A.~C. Davison, and S.~Engelke.
\newblock Extremes on river networks.
\newblock \emph{Annals of Applied Statistics}, 9\penalty0 (4):\penalty0
  2023--2050, 2015.

\bibitem[Asadi et~al.(2018)Asadi, Engelke, and Davison]{asadiregio}
P.~Asadi, S.~Engelke, and A.~C. Davison.
\newblock Optimal regionalization of extreme value distributions for flood
  estimation.
\newblock \emph{Journal of Hydrology}, 556:\penalty0 182--193, 2018.

\bibitem[Beirlant et~al.(2004)Beirlant, Goegebeur, Segers, and
  Teugels]{beirlant}
J.~Beirlant, Y.~Goegebeur, J.~Segers, and J.~Teugels.
\newblock \emph{Statistics of Extremes: Theory and Applications}.
\newblock Wiley, Chichester, 2004.

\bibitem[Bodik et~al.(2024)Bodik, Palu{\v{s}}, and Pawlas]{bodik}
J.~Bodik, M.~Palu{\v{s}}, and Z.~Pawlas.
\newblock Causality in extremes of time series.
\newblock \emph{Extremes}, 27:\penalty0 67--121, 2024.

\bibitem[Chautru(2015)]{Chautru}
E.~Chautru.
\newblock Dimension reduction in multivariate extreme value analysis.
\newblock \emph{Electronic Journal of Statistics}, 9:\penalty0 383--418, 2015.

\bibitem[Cooley and Thibaud(2019)]{cooley}
D.~Cooley and E.~Thibaud.
\newblock Decompositions of dependence for high-dimensional extremes.
\newblock \emph{Biometrika}, 106\penalty0 (3):\penalty0 587--604, 2019.

\bibitem[{de~Haan} and Ferreira(2006)]{DHF}
L.~{de~Haan} and A.~Ferreira.
\newblock \emph{Extreme Value Theory: An Introduction}.
\newblock Springer, New York, 2006.

\bibitem[Drton and Maathuis(2017)]{drtonreview}
M.~Drton and M.~H. Maathuis.
\newblock Structure learning in graphical modeling.
\newblock \emph{Annual Review of Statistics and Its Application}, 4:\penalty0
  365--393, 2017.

\bibitem[Einmahl et~al.(2012)Einmahl, Krajina, and Segers]{einmahl2012m}
J.~H.~J. Einmahl, A.~Krajina, and J.~Segers.
\newblock An {M}-estimator for tail dependence in arbitrary dimensions.
\newblock \emph{Annals of Statistics}, 40\penalty0 (3):\penalty0 1764--1793,
  2012.

\bibitem[Engelke and Hitz(2020)]{engelke:hitz:18}
S.~Engelke and A.~S. Hitz.
\newblock Graphical models for extremes (with discussion).
\newblock \emph{Journal of the Royal Statistical Society, Series B},
  82\penalty0 (4):\penalty0 871--932, 2020.

\bibitem[Engelke and Volgushev(2022)]{ES}
S.~Engelke and S.~Volgushev.
\newblock Structure learning for extremal tree models.
\newblock \emph{Journal of the Royal Statistical Society, Series B},
  84\penalty0 (5):\penalty0 2055--2087, 2022.

\bibitem[Engelke et~al.(2022)Engelke, Lalancette, and Volgushev]{eglearn}
S.~Engelke, M.~Lalancette, and S.~Volgushev.
\newblock Learning extremal graphical structures in high dimensions.
\newblock arXiv:2111.00840, 2022.

\bibitem[Foug\`eres et~al.(2013)Foug\`eres, Mercadier, and Nolan]{foug}
A.~L. Foug\`eres, C.~Mercadier, and J.~P. Nolan.
\newblock Dense classes of multivariate extreme value distributions.
\newblock \emph{Journal of Multivariate Analysis}, 116:\penalty0 109--129,
  2013.

\bibitem[Friedman et~al.(2008)Friedman, Hastie, and Tibshirani]{glasso}
J.~Friedman, T.~Hastie, and R.~Tibshirani.
\newblock Sparse inverse covariance estimation with the graphical lasso.
\newblock \emph{Biostatistics}, 9\penalty0 (3):\penalty0 432--441, 2008.

\bibitem[Gissibl and Kl\"uppelberg(2018)]{gk}
N.~Gissibl and C.~Kl\"uppelberg.
\newblock Max-linear models on directed acyclic graphs.
\newblock \emph{Bernoulli}, 24\penalty0 (4A):\penalty0 2693--2720, 2018.

\bibitem[Gnecco et~al.(2021)Gnecco, Meinshausen, Peters, and Engelke]{gnecco}
N.~Gnecco, N.~Meinshausen, J.~Peters, and S.~Engelke.
\newblock Causal discovery in heavy-tailed models.
\newblock \emph{Annals of Statistics}, 49\penalty0 (3):\penalty0 1755--1778,
  2021.

\bibitem[Gnecco et~al.(2023)Gnecco, Terefe, and Engelke]{exforest}
N.~Gnecco, E.~M. Terefe, and S.~Engelke.
\newblock Extremal random forests.
\newblock \emph{Journal of the American Statistical Association}, 2023.
\newblock URL \url{https://doi.org/10.1080/01621459.2023.2300522}.

\bibitem[Goix et~al.(2017)Goix, Sabourin, and Cl{\'e}men{\c c}on]{goix}
N.~Goix, A.~Sabourin, and S.~Cl{\'e}men{\c c}on.
\newblock Sparse representation of multivariate extremes with applications to
  anomaly detection.
\newblock \emph{Journal of Multivariate Analysis}, 161:\penalty0 12--31, 2017.

\bibitem[Heffernan and Tawn(2004)]{hefftawn}
J.~E. Heffernan and J.~A. Tawn.
\newblock A conditional approach for multivariate extreme values (with
  discussion).
\newblock \emph{Journal of the Royal Statistical Society, Series B},
  66\penalty0 (3):\penalty0 497--546, 2004.

\bibitem[Hyv{\"a}rinen and Smith(2013)]{pairwiselingam}
A.~Hyv{\"a}rinen and S.~M. Smith.
\newblock Pairwise likelihood ratios for estimation of non-gaussian structural
  equation models.
\newblock \emph{Journal of Machine Learning Research}, 14\penalty0
  (1):\penalty0 111--152, 2013.

\bibitem[Jan{\ss}en and Wan(2020)]{JanWan}
A.~Jan{\ss}en and P.~Wan.
\newblock $ k $-means clustering of extremes.
\newblock \emph{Electronic Journal of Statistics}, 14\penalty0 (1):\penalty0
  1211--1233, 2020.

\bibitem[Kl\"uppelberg and Krali(2021)]{KK}
C.~Kl\"uppelberg and M.~Krali.
\newblock Estimating an extreme {B}ayesian network via scalings.
\newblock \emph{Journal of Multivariate Analysis}, 181\penalty0 (1), 2021.
\newblock \doi{https://doi.org/10.1016/j.jmva.2020.104672}.

\bibitem[Krali et~al.(2023)Krali, Davison, and Kl\"uppelberg]{KDK}
M.~Krali, A.~C. Davison, and C.~Kl\"uppelberg.
\newblock Heavy-tailed max-linear structural equation models in networks with
  hidden nodes.
\newblock arXiv:2306.15356, 2023.

\bibitem[Larsson and Resnick(2012)]{lars}
M.~Larsson and S.~I. Resnick.
\newblock Extremal dependence measure and extremogram: the regularly varying
  case.
\newblock \emph{Extremes}, 15:\penalty0 231--256, 2012.

\bibitem[Lauritzen(1996)]{lau}
S.~L. Lauritzen.
\newblock \emph{Graphical Models}.
\newblock Clarendon Press, Oxford, 1996.

\bibitem[Lee and Cooley(2022)]{leecooley}
J.~Lee and D.~Cooley.
\newblock Partial tail correlation for extremes.
\newblock arXiv: 2210.02048, 2022.

\bibitem[Meyer and Wintenberger(2023)]{meyerclust}
N.~Meyer and O.~Wintenberger.
\newblock Multivariate sparse clustering for extremes.
\newblock \emph{Journal of the American Statistical Association}, 119\penalty0
  (547):\penalty0 1911--1922, 2023.

\bibitem[Mhalla et~al.(2020)Mhalla, Chavez-Demoulin, and Dupuis]{mhalla}
L.~Mhalla, V.~Chavez-Demoulin, and D.~J. Dupuis.
\newblock Causal mechanism of extreme river discharges in the upper {D}anube
  basin network.
\newblock \emph{Journal of the Royal Statistical Society, Series C},
  69\penalty0 (4):\penalty0 741--764, 2020.

\bibitem[Pasche et~al.(2023)Pasche, Chavez-Demoulin, and Davison]{PCD}
O.~C. Pasche, V.~Chavez-Demoulin, and A.~C. Davison.
\newblock Causal modelling of heavy-tailed variables and confounders with
  application to river flow.
\newblock \emph{Extremes}, 26:\penalty0 573--594, 2023.

\bibitem[Pearl(2009)]{pearl}
J.~Pearl.
\newblock \emph{Causality: Models, Reasoning, and Inference}.
\newblock Cambridge University Press, Cambridge, second edition, 2009.

\bibitem[Peters and B{\"u}hlmann(2015)]{sid}
J.~Peters and P.~B{\"u}hlmann.
\newblock Structural intervention distance for evaluating causal graphs.
\newblock \emph{Neural computation}, 27\penalty0 (3):\penalty0 771--799, 2015.

\bibitem[Resnick(1987)]{sres}
S.~I. Resnick.
\newblock \emph{Extreme Values, Regular Variation, and Point Processes}.
\newblock Springer, New York, 1987.

\bibitem[Resnick(2007)]{ResnickHeavy}
S.~I. Resnick.
\newblock \emph{Heavy-Tail Phenomena: Probabilistic and Statistical Modeling}.
\newblock Springer, New York, 2007.

\bibitem[Shimizu et~al.(2006)Shimizu, Hoyer, Hyv{\"a}rinen, Kerminen, and
  Jordan]{lingam1}
S.~Shimizu, P.~O. Hoyer, A.~Hyv{\"a}rinen, A.~Kerminen, and M.~Jordan.
\newblock A linear non-{G}aussian acyclic model for causal discovery.
\newblock \emph{Journal of Machine Learning Research}, 7:\penalty0 2003--2030,
  2006.

\bibitem[Shimizu et~al.(2011)Shimizu, Inazumi, Sogawa, Hyvarinen, Kawahara,
  Washio, Hoyer, and Bollen]{directlingam}
S.~Shimizu, T.~Inazumi, Y.~Sogawa, A.~Hyvarinen, Y.~Kawahara, T.~Washio, P.~O.
  Hoyer, and K.~Bollen.
\newblock {DirectLiNGAM}: A direct method for learning a linear non-{G}aussian
  structural equation model.
\newblock \emph{Journal of Machine Learning Research}, 12:\penalty0 1225--1248,
  2011.

\bibitem[Spirtes et~al.(2000)Spirtes, Glymour, and Scheines]{spirt}
P.~Spirtes, C.~Glymour, and R.~Scheines.
\newblock \emph{Causation, Prediction, and Search}.
\newblock MIT Press, Cambridge, MA, second edition, 2000.

\bibitem[Tran et~al.(2024)Tran, Buck, and Kl\"uppelberg]{TBK}
N.~Tran, J.~Buck, and C.~Kl\"uppelberg.
\newblock Estimating a directed tree for extremes.
\newblock \emph{Journal of the Royal Statistical Society, Series B},
  86\penalty0 (3):\penalty0 771--792, 2024.

\bibitem[Wang and Stoev(2011)]{Wang2011}
Y.~Wang and S.~A. Stoev.
\newblock Conditional sampling for spectrally discrete max-stable random
  fields.
\newblock \emph{Advances in Applied Probability}, 43\penalty0 (2):\penalty0
  461--483, 2011.

\bibitem[Wang and Drton(2020)]{wangcausal}
Y.~S. Wang and M.~Drton.
\newblock High-dimensional causal discovery under non-gaussianity.
\newblock \emph{Biometrika}, 107\penalty0 (1):\penalty0 41--59, 2020.

\end{thebibliography}
\bibliographystyle{plainnat}

\newpage

\hspace{-.575cm}{\huge{\textbf{Appendix}}}

\begin{appendices}
\section{Multivariate regular variation}\label{p2sec:ARV}

\subsection{Definitions and results for regularly varying vectors}\label{p2sec:3.1}

We use two equivalent definitions of multivariate regular variation from Theorem 6.1 of \citet{ResnickHeavy}.

\begin{definition} \label{p2mrv} 
	(i)\, A random vector $\boldsymbol{X}\in\mathbb{R}^{d}_+$ is {multivariate regularly varying} if there exists a sequence $b_n\to \infty$  as $n\to\infty$ such that
\begin{align}\label{p2eq:mrva}
n\mathbb{P}(\boldsymbol{X}/{b_n}\in \cdot)\overset{v}{\to}{\nu_{\boldsymbol{X}}}(\cdot), \hspace{5mm}n\to\infty,
\end{align}
where $\overset{v}{\to}$ denotes vague convergence in $M_+([0, \infty]^{d}\setminus\{\boldsymbol{0}\})$, the set of non-negative Radon measures on $[0, \infty]^{d}\setminus\{\boldsymbol{0}\}$, and $\nu_{\boldsymbol{X}}$ is called the {exponent measure} of $\boldsymbol X$.\\
(ii)\,
A random vector  $\boldsymbol{X}\in\mathbb{R}^{ {d}}_+$ is {multivariate regularly varying} if for any norm $\|\cdot\|$ there exists a finite measure $H_{\boldsymbol{X}}$ on the positive unit sphere $\Theta_+^{{ {d}}-1}=\{\boldsymbol{\omega}\in \mathbb{R}^{ {d}}_+: \norm{\boldsymbol{\omega}}=1\}$ and a sequence $b_n\to \infty$ as $n\to\infty$ such that the angular representation $(R,\boldsymbol{\omega})\coloneqq(\norm{\boldsymbol{X}}, \boldsymbol{X}/\norm{\boldsymbol{X}})$ of $\boldsymbol X$ satisfies
\begin{align}\label{p2eq:mrvb}
n\mathbb{P}\left(\left({R}/{b_n},\boldsymbol{\omega}\right)\in \cdot\right)\overset{v}{\to} \nu_\alpha\times H_{\boldsymbol{X}}(\cdot), \hspace{5mm}n\to\infty,
\end{align}
in $M_+((0,\infty]\times\Theta_+^{{ {d}}-1})$, ${\rm d}\nu_\alpha(x)=\alpha x^{-\alpha-1}{\rm d}x$ for some $\alpha>0$, and for Borel subsets $C\subseteq \Theta_+^{{ {d}}-1}$,
\begin{align*}
H_{\boldsymbol{X}}(C)\coloneqq\nu_{\boldsymbol{X}}\big(\{\boldsymbol{x}\in\mathbb{R}^{ {d}}_+\setminus\{\boldsymbol{0}\}: \norm{\boldsymbol{x}}\geq 1, \boldsymbol{x}/\norm{\boldsymbol{x}} \in {C}\}\big).
\end{align*}
The measure $H_{\boldsymbol{X}}$ is called the {angular measure of $\boldsymbol{X}$}, we write $\boldsymbol{X}\in {\rm RV}^{ {d}}_+(\alpha)$, and $\alpha$ is called the {index of regular variation}.
\end{definition}
	
The angular measure	$H_{\boldsymbol{X}}$ may not be a probability measure because it carries information on the scalings of the components of $\bs X$. This motivates its normalisation.
	
\brem
(i) \, \label{p2specmass2} As $H_{\bsx}$ is a finite measure, it can be normalised to a probability measure by defining 
$$\tilde{H}_{\bsx}(\cdot)\coloneqq \frac{H_{\boldsymbol{X}}(\cdot)}{H_{\boldsymbol{X}}(\Theta_+^{{d}-1})}.$$ 
(ii) \, \label{p2enx} Let $(R, \bs \omega)$ denote the angular representation of $\bs X \in {\rm RV^d_+(\alpha)}$ as in Definition~\ref{p2mrv} (b).
Let $f\colon\Theta_+^{{d}-1}\to\mathbb{R}_+$ be a continuous function. Since $f$ is compactly supported (on $\Theta_+^{{d}-1}$), by vague convergence we have 
\begin{align}\label{p2empdist}
\mathbb{E}_{\tilde{H}_{\bsx}}[f(\boldsymbol{\omega})]& \coloneqq\lim\limits_{x\to \infty} \mathbb{E}[f(\boldsymbol{\omega})\mid R>x] = \int_{\Theta_+^{{d}-1}} f(\bs{\omega}) d\tilde{H}_{\bsx}(\bs{\omega}).
\end{align}
\erem

\subsection{Regularly varying linear structural equation models}\label{p2sec:3.2}

	In this section we mostly follow \citet{ResnickHeavy}, but the results are similar to those in \citet[Section 6]{einmahl2012m} or \citet{cooley}.
The latter consider linear operations between extremes, but apply an additional transformation to map observations to the positive orthant.

	Initially, for a given component $Z_i$ of $\boldsymbol{Z}$, for simplicity denoted by $Z$, we consider the vector $\boldsymbol{a} Z=(a_1Z,\ldots,a_dZ)$, where $a_i\geq 0$ for $i\in\{1,\ldots,d\}$. 
	
	\begin{lemma}\label{p2a.z}
		Let ${Z}\in {\rm RV}_+(\alpha)$ be a standardised random variable, and let $\boldsymbol{a}\in\mathbb{R}^d_+$. Then $\boldsymbol{a} Z$ has exponent measure 
		\begin{align*}
		\nu_{\boldsymbol{a} Z}([\boldsymbol{0},\boldsymbol{x}]^c)=\underset{i\in\{1,\ldots,d\}}{\max}\,\frac{a_i^{\alpha}}{x_i^{\alpha}},\quad \boldsymbol{x}\in \mathbb{R}_+^d\setminus\{\bs 0\}.
		\end{align*}
	\end{lemma}
	\begin{proof}
		From Lemma 6.1 in \cite{ResnickHeavy},  convergence on the rectangles $[\boldsymbol{0},\boldsymbol{x}]^c$ for $\boldsymbol{x}\in \mathbb{R}_+^d$ which are continuity points of $\nu_{\boldsymbol{a} Z}$ is equivalent to convergence in $M_+([0, \infty]^{ {d}}\setminus\{\boldsymbol{0}\})$. Therefore it suffices to consider only such rectangular sets, for which
		\begin{align*}
		n\mathbb{P}\left(\boldsymbol{a}Z/b_n\in [\boldsymbol{0},\boldsymbol{x}]^c\right)&=n\mathbb{P}\left(\overset{d}{\underset{i=1}{\cup}}\{a_iZ>b_nx_i\}\right)\\
		&=n\mathbb{P}\left(\overset{d}{\underset{i=1}{\cup}}\left\{Z>b_n\frac{x_i}{a_i}\right\}\right)\\
		&=n\mathbb{P}\left(Z>b_n\underset{i\in\{1,\ldots,d\}}{\min}\,\frac{x_i}{a_i}\right)\\
		&=\left(\underset{i\in\{1,\ldots,d\}}{\min}\,\frac{x_i}{a_i}\right)^{-\alpha}=\underset{i\in\{1,\ldots,d\}}{\max}\,\frac{a_i^{\alpha}}{x_i^{\alpha}},
		\end{align*}
		proving the claim.
	\end{proof}
	We turn our attention to the following result from  \citet[Proposition 7.4]{ResnickHeavy}.
	\begin{proposition}\label{p2sums}
		Let the independent vectors $\boldsymbol{X}, \boldsymbol{Y}\in {\rm RV}^d_+(\alpha)$ be defined on the same probability space and with the same normalising sequence $b_n$. Then
		\begin{align}
		n\mathbb{P}\left(\frac{\boldsymbol{X}+\boldsymbol{Y}}{b_n}\in \cdot\right)\overset{v}{\to}\nu_{\boldsymbol{X}}(\cdot)+\nu_{\boldsymbol{Y}}(\cdot)
		\end{align}
		in $M_+([0, \infty]^{d}\setminus\{\boldsymbol{0}\})$.
	\end{proposition}
	
	Using Lemma~\ref{p2a.z} and Proposition \ref{p2sums}, we obtain the following.
	
	\begin{proposition}\label{p2discspect}
		Consider $A\in \mathbb{R}_+^{d\times d}$ and $\boldsymbol{Z}\in {\rm RV}^d_+(\alpha)$. Then $A\boldsymbol{Z}$ is multivariate regularly varying 
		and its angular measure on ${\Theta_+^{d-1}}$ equals
		\begin{align*}
		{H}_{A\boldsymbol{Z}}(\cdot)=\sum_{j=1}^d\norm{\boldsymbol{a}_j}^\alpha\delta_{\big\{\frac{\boldsymbol{a}_j}{\norm{\boldsymbol{a}_j}}\big\}}(\cdot),
		\end{align*}
		where $\boldsymbol{a}_j$ is the $j$-th column of the matrix $A$.
		\begin{proof}
			We may write 
			\begin{align*}
			A\boldsymbol{Z}=\sum_{j=1}^{d}(\boldsymbol{a}_jZ_j)^\top.
			\end{align*}
			As in Lemma~\ref{p2a.z}, we need only prove convergence on the rectangles $[\boldsymbol{0},\boldsymbol{x}]^c$ for $\boldsymbol{x}\in \mathbb{R}_+^d\setminus\{\bs 0\}$. Proposition \ref{p2sums} applied to the vectors $\boldsymbol{a}_j Z_j$ gives
			\begin{align*}
			\nu_{A\boldsymbol{Z}}([\boldsymbol{0},\boldsymbol{x}]^c)=\sum_{j=1}^{d}\underset{i\in\{1,\ldots,d\}}{\max}\,\frac{a_{ij}^{\alpha}}{x_i^{\alpha}}.
			\end{align*}
			
			The angular decomposition from Definition~\ref{p2mrv}(b) gives
			\begin{align*}
			\nu_{A\boldsymbol{Z}}([\boldsymbol{0},\boldsymbol{x}]^c)=\int_{\Theta_+^{d-1}}\underset{i\in\{1,\ldots,d\}}{\bigvee}\left(\frac{x_i}{\omega_i}\right)^{-\alpha}H_{A\boldsymbol{Z}}({\rm{d} \bs{\omega}}),
			\end{align*}
and it follows that 	${H}_{A\boldsymbol{Z}}(\cdot)=\sum_{j=1}^d\norm{\boldsymbol{a}_j}^\alpha\delta_{\left\{{\boldsymbol{a}_j}/{\norm{\boldsymbol{a}_j}}\right\}}(\cdot)$.
		\end{proof}
	\end{proposition}

	\subsection{Scalings}\label{p2scalings}
	In this section we work under Assumptions A and consider the LSEM+ vector $\boldsymbol{X}=A\boldsymbol{Z}$ with $\bs Z \in {\rm RV}^d_+(2)$ and with exponent and angular measures $\nu_{\bsx}$ and $H_{\bsx}$. We use Definition \ref{p2mrv} (b) and recall that the scalings for the $i$-th margin of $\bsx$ can be expressed as
	\begin{align*}
	\lim\limits_{n\to \infty} n\mathbb{P}({{X_i}}/b_{n}>x)&=\nu_{\boldsymbol{X}}\left({\left\{\boldsymbol{x}\in\R_+^{ {d}} : \frac{\boldsymbol{x}}{\norm {\boldsymbol{x}}}\in \Theta_+^{{d}-1}, x_i > x\right\}}\right)\\
	&=\int_{\{\boldsymbol{\omega}\in \Theta_+^{ {d}-1}\}} \int_{\{r> x/{\omega_i}\}} 2 r^{-3}{\rm d}r  {\rm d}H_{{\boldsymbol{X}}}(\boldsymbol{\omega})\\
	&=x^{-2}\int_{\{\boldsymbol{\omega}\in \Theta_+^{ {d}-1}\}} \omega_i^{2} {\rm d}H_{{\boldsymbol{X}}}(\boldsymbol{\omega})\\
	&=x^{-2}\sigma_{i}^2,
	\end{align*}
	with $\sigma_{i}^2$ obtained by setting $x=1$.  Proposition~\ref{p2discspect}  gives 
	\begin{align*}
	\sigma_{i}^2 =\int_{\Theta_+^{{d}-1}}\omega_i^2 {\rm d}H_{\boldsymbol{X}}(\boldsymbol{\omega})
	= \sum_{k=1}^{d}\norm{\bs a_k}^2 \frac{a_{ik}^2}{\norm{\bs a_k}^2}=\sum_{k=1}^{d}a_{ik}^2,\quad i\in\{1,\ldots,d\}.
	\end{align*}

	Following the discussion on page 330 of \cite{beirlant}, who assume standardisation with $\alpha=1$, a particularly useful choice of the norm is $\norm{\cdot}_\alpha$ (without loss of generality for $\alpha\geq 1$). Taking $\alpha=2$ and the Euclidean norm $\norm{\cdot}$ gives
	\begin{align}\label{p2specmass}
	H_{\boldsymbol{X}}(\Theta_+^{{d}-1}) =\int_{\Theta_+^{{d}-1}}\sum_{i=1}^{d}\omega_i^2 {\rm d}H_{\boldsymbol{X}}(\boldsymbol{\omega})
	= \sum_{i=1}^{d}\sigma_{i}^2 =\sum_{i=1}^{d}\sum_{j=1}^{d}a_{ij}^2=d,
	\end{align}
	which follows from Assumption (A3) on the standardised rows of $A$ by the $\norm{\cdot}$ norm.

	Finally, for $M_I=\max(X_i: i\in I)$,  Proposition \ref{p2discspect} gives
	\begin{align*}
	\sigma_{M_{\boldsymbol{I}}}^2 =\int_{\Theta_+^{{d}-1}}\underset{i\in I}{\vee}\omega_i^{2} {\rm d}H_{\boldsymbol{X}}(\boldsymbol{\omega})
	= \sum_{k=1}^{d}\norm{a_k}^2\left( \underset{i\in I}{\vee}\frac{a_{ik}^2}{\norm{a_k}^2}\right)=\sum_{k=1}^{d}\underset{i\in I}{\vee}a_{ik}^2.
	\end{align*}
	
	An application of Lemma~\ref{p2ineq} to the computed scalings gives Lemma \ref{p2scalcoll}.

\section{Proofs for Sections \ref{p2prelims} and \ref{p2slearn}}	\label{p2proofs}
	
	\begin{proof}[\textbf{Proof of Proposition \ref{p2allpathcond}.}]
		We mimic the steps used in the proof of Theorem 3.10 of \cite{gk}, but adjust for linear operations.  
		
		By (\ref{p2lsemcoef}) and the definition of path weights we have that $d(p_{jk})=s_{jj}c_{jl_1}\cdots c_{l_{v-1}i}$ for some path length $v$. Likewise we obtain $d(p_{ki})$ for the path $k\rightsquigarrow i$. Then, for the path $j\rightsquigarrow k \rightsquigarrow i$, composed of $p_{jk}$ and $p_{ki}$ and denoted by $[p_{jk},p_{ki}]$, we compute the path weight $$d([p_{jk},p_{ki}])=d(p_{jk})d(p_{ki})/s_{kk}=d(p_{jk})d(p_{ki})/a_{kk}.$$
		For ease of notation we write $\mathcal{K}_{ji}$ for the set of paths from $i$ to $j$ that pass through $k$, and $\mathcal{K}_{ji}^c$ for those which do not. We note that we can also set the node indices $i$ or $j$ to $k$.
		These two sets form a partition of the paths from $j$ to $i$, so we may write
		\begin{align*}
		a_{ij}=\underset{p_{ji}\in \mathcal{K}_{ji}}{\sum}d(p_{ji})+{\underset{p_{ji}\in \mathcal{K}_{ji}^c}{\sum}d(p_{ji})}&=\underset{\{p_{jk}\in\mathcal{K}_{jk},\, p_{ki}\in \mathcal{K}_{ki}\}}{\sum}d([p_{jk},p_{ki}])+{\underset{p_{ji}\in \mathcal{K}_{ji}^c}{\sum}d(p_{ji})}\\
		&=\underset{\{p_{jk}\in\mathcal{K}_{jk},\, p_{ki}\in \mathcal{K}_{ki}\}}{\sum}d(p_{jk})d(p_{ki})/a_{kk}+{\underset{p_{ji}\in \mathcal{K}_{ji}^c}{\sum}d(p_{ji})}\\
		&=a_{kk}^{-1}\underset{p_{jk}\in\mathcal{K}_{jk}}{\sum}d(p_{jk})\underset{p_{ki}\in \mathcal{K}_{ki}}{\sum}d(p_{ki})+{\underset{p_{ji}\in \mathcal{K}_{ji}^c}{\sum}d(p_{ji})}\\
		&=a_{kj}a_{ik}/a_{kk}+{\underset{p_{ji}\in \mathcal{K}_{ji}^c}{\sum}d(p_{ji})}.
		\end{align*}
		
		Now suppose that $a_{ij}>a_{ik}a_{kj}/a_{kk}$. Because $C$ and $A$ are non-negative, this can happen if and only if $\mathcal{K}_{ji}^c$ is non-empty. Otherwise, $a_{ij}=a_{ik}a_{kj}/a_{kk}$, thus proving the proposition. 
	\end{proof}
	
	\begin{proof}[\textbf{Proof of Lemma \ref{p2ineq}.}] First note that the set  $\An(i)\setminus\An(j)$ contains $i$ and thus is non-empty. We then obtain
		$a_{jj}^\alpha\sum_{k\in \An(i)\setminus\An(j)}a_{ik}^\alpha>0$, which is equivalent to
		\begin{align*}
		a_{jj}^\alpha \sum_{l\in\An(j)}\frac{a_{ij}^\alpha a_{jl}^\alpha}{a_{jj}^\alpha}+a_{jj}^\alpha\sum_{k\in \An(i)\setminus\An(j)}a_{ik}^\alpha&>a_{ij}^\alpha\sum_{l\in\An(j)} a_{jl}^\alpha,
		\end{align*}
		and re-arranging the terms gives 
		\begin{align*}
		\frac{a_{jj}^\alpha}{\sum_{l\in\An(j)} a_{jl}^\alpha}&>\frac{a_{ij}^\alpha}{\sum_{l\in\An(j)}\frac{a_{ij}^\alpha a_{jl}^\alpha}{a_{jj}^\alpha}+\sum_{k\in \An(i)\setminus\An(j)}a_{ik}^\alpha}.
		\end{align*}
		We apply Corollary~\ref{p2cor1} to the inequality in the last display and find
		\begin{align*}
		\frac{a_{jj}^\alpha}{\sum_{l\in\An(j)} a_{jl}^\alpha}&>\frac{a_{ij}^\alpha}{\sum_{l\in\An(j)}{a_{il}^\alpha}+\sum_{k\in \An(i)\setminus\An(j)}a_{ik}^\alpha},
		\end{align*}
		which corresponds to
		$\bar{a}_{jj}^\alpha>\bar{a}_{ij}^\alpha$ or equivalently $\bar{a}_{jj}>\bar{a}_{ij}$,
		and thus proves the claim of the lemma.
	\end{proof}

		\begin{proof}[\textbf{Proof of Theorem \ref{p2sourcenodes}.}]
			($\Rightarrow$) Suppose that $j$ is a source node.   Then $\an(j)=\emptyset$ by definition, so $a_{jk}=0$ for all $k\neq j$, and $a_{jj}=1$.
			
			Consider $M_{ij}$ for some $i\neq j$. By Lemma~\ref{p2scalcoll}, its squared scaling is 
			\[ \sigma_{M_{ij}}^2=\sum_{k\in \An(i)\cup\an(j)\setminus\{j\}} a_{ik}^2+a_{ij}^2\vee a_{jj}^2=\sum_{k\in \An(i)\cup\an(j)\setminus\{j\}} a_{ik}^2+  a_{jj}^2.
			\]
			Choose $a>1$, and consider $M_{i,aj}$. A second application of Lemma~\ref{p2scalcoll} gives
			 \[\sigma_{M_{i,aj}}^2=\sum_{k\in \An(i)\cup\an(j)\setminus\{j\}} a_{ik}^2+  a^2 a_{jj}^2.\] 
			Finally, we compute the difference $$\sigma_{M_{i,aj}}^2-\sigma_{M_{ij}}^2=(a^2-1)  a_{jj}^2=a^2-1.$$
			
			($\Leftarrow$) $\hspace{5mm}$ For the other direction we argue by contradiction.  Suppose that $j$ is not a source node, so there exists a node $i$ such that $i\in \an(j)$ and $a_{ji}>0$, and that equality \eqref{p2critinit} holds. Lemma~\ref{p2ineq} implies that $a_{ji}<a_{ii}$. We proceed similar to the first part of the proof and compute 
			\[ \sigma_{M_{ij}}^2=\sum_{k\in \An(j)\cup\an(i)\setminus\{i\}} a_{jk}^2\vee a_{ik}^2+a_{ji}^2\vee a_{ii}^2=\sum_{k\in \An(j)\cup\an(i)\setminus\{i\}} a_{jk}^2\vee a_{ik}^2+  a_{ii}^2.
			\]
			
			Likewise,  the squared scaling for $M_{i,aj}$ equals
			\begin{align*}\sigma_{M_{i,aj}}^2=\sum_{k\in \An(i)\cup\an(i)\setminus\{i\}} a^2\,a_{jk}^2\vee a_{ik}^2+a^2\,a_{ji}^2\vee a_{ii}^2.
			\end{align*}
			
			We first define the sets $S_1(i,j)=\{k\in\An(i)\cup\An(j):  a_{jk}\geq a_{ik}\}$ and $S_2(i,j)=\{k\in \An(i)\cup\An(j):  a_{jk}<a_{ik} \,\,{\text{and}}\,\, a\, a_{jk}\geq a_{ik}\}$. Then 
			\begin{align*} \sigma_{M_{i,aj}}^2-\sigma_{M_{ij}}^2&=\sum_{k\in S_1(i,j)} (a^2-1) a_{jk}^2+\sum_{k\in S_2(i,j)} (a^2\,a_{jk}^2-a_{ik}^2)\\
			&< \sum_{k\in S_1(i,j)} (a^2-1) a_{jk}^2+\sum_{k\in S_2(i,j)} (a^2-1)\,a_{jk}^2\\
			&\leq \sum_{\substack{k\in \An(i)\cup\An(j)}}  (a^2-1)\,a_{jk}^2=(a^2-1)\sigma_{M_{j}}^2=(a^2-1),
			\end{align*} 
			where the inequality from the first to the second lines follows by an argument similar to \eqref{p2exineq} in Example \ref{p2theo1ex1}, because the index $k=i$ lies in the set $S_2(i,j)$ in the second summand. We have a contradiction of \eqref{p2critinit}, so $j$ cannot be a source node.
		\end{proof}

		\begin{proof}[\textbf{Proof of Theorem \ref{p2descnodes}.}]
			($\Rightarrow$) Suppose that node $j$ is such that $\an(j)\cap I^c=\emptyset$. Choose an arbitrary node $i\notin I\cup\{j\}$ and consider the squared scalings of $M_{i, j, I}$, and $M_{i, aj, aI}$. Note that  since $I$ has no ancestors outside $I$, $a_{rk}=0$ for all nodes $r\in I$ and $k\in I^c$.
			
			Using the properties of $I$, Lemmas~\ref{p2ineq} and~\ref{p2scalcoll}, and following the steps in the proof of Theorem~\ref{p2sourcenodes}, we compute
			\begin{align*}
			\sigma_{M_{i, aj, aI}}^2&=a^2\sum_{k\in I}a_{kk}^2+\sum_{k\in I^c\cap \An(j)}a^2\,a_{jk}^2 +\sum_{k\in I^c\cap \An(j)^c\cap \An(i)}a_{ik}^2,\nonumber\\
			\sigma_{M_{i, j, I}}^2&=\sum_{k\in I}a_{kk}^2+\sum_{k\in I^c\cap \An(j)}a_{jk}^2 +\sum_{k\in I^c\cap \An(j)^c\cap \An(i)}a_{ik}^2,\nonumber\\
			\sigma_{M_{i, aj, aI}}^2-	\sigma_{M_{i, j, I}}^2&=(a^2-1)\left(\sum_{k\in I}a_{kk}^2+\sum_{k\in I^c\cap \An(j)}a_{jk}^2\right)=(a^2-1)\sigma_{M_{j, I}}^2,
			\end{align*}
			where the last line corresponds to~\eqref{p2critpair}.
			
			($\Leftarrow$)  Suppose now that~\eqref{p2critpair} holds  and that there exists $i\in I^c\cap\an(j)$. For notational simplicity, we define the sets $S_1(i,j,I)=\{k\in \left(\An(j)\cup\An(i)\right)\cap I^c:  a_{jk}\geq a_{ik}\}$ and $S_2(i,j,I)=\{k\in \left(\An(j)\cup\An(i)\right)\cap I^c:  a_{jk}<a_{ik} \,\,{\text{and}}\,\, a\,a_{jk}\geq a_{ik}\}$. We then compute
			\begin{align*}
			\sigma_{M_{i, aj, aI}}^2-	\sigma_{M_{i, j, I}}^2-(a^2-1)\sum_{k\in I}a_{kk}^2&=\sum_{ k\in S_1(i,j,I)} (a^2-1) a_{jk}^2+\sum_{k\in S_2(i,j,I)} (a^2\,a_{jk}^2-a_{ik}^2)\\
			&< \sum_{ k\in S_1(i,j,I)}  (a^2-1) a_{jk}^2+\sum_{k\in S_2(i,j,I)} (a^2-1)\,a_{jk}^2\\
			&\leq \sum_{\substack{k\in \An(j)\cap I^c}}  (a^2-1)\,a_{jk}^2,
			\end{align*}
			where the first equality follows because of Lemmas~\ref{p2ineq} and~\ref{p2scalcoll}, and properties of the set $I$.  The inequality in the seond line follows by Lemmas~\ref{p2ineq} and~\ref{p2scalcoll}, and the inequality $a_{ii}>a_{ji}$ for $i\in S_2(i, j, I)$.
			Rearranging the terms in the first and last line of the last display we obtain
			\begin{align*}
			\sigma_{M_{i, aj, aI}}^2-	\sigma_{M_{i, j, I}}^2<(a^2-1)\sum_{k\in I}a_{kk}^2+\sum_{\substack{k\in \An(j)\cap I^c}}  (a^2-1)\,a_{jk}^2= (a^2-1)\sigma_{M_{j,I}}^2,
			\end{align*}
			which contradicts the initial equality~\eqref{p2critpair}. Hence $j$ cannot be such that $\an(j)\cap I^c=\emptyset$.
		\end{proof}

		\section{The empirical angular measure}\label{p2estprocedure}
		In this section we outline the statistical theory used in estimating the scaling parameters in Sections \ref{p2slearn} and \ref{p2simdata}. The material starts from Section 9.2 in \citet{ResnickHeavy}, and is similar to  Section 6 of \cite{KK}. We write $\stD$ to denote convergece in distribution,~$\stp$ for convergence in probability and $\stw$ for weak convergence.
		
	Given $n$ independent replicates $\boldsymbol{X}_1,\dots,\boldsymbol{X}_n$ of $\boldsymbol{X}\in {\rm RV}^{d}_+(2)$, we compute
	\begin{align}\label{p2pol.est}
	R_\ell\coloneqq\norm{\boldsymbol{X}_\ell}_2,\quad\quad {\boldsymbol{\omega}_\ell=({\omega}_{\ell1},\ldots,{\omega}_{\ell {d}})\coloneqq\frac{\boldsymbol{X}_\ell}{R_\ell}},\quad\quad \ell\in\{1,\dots,n\},
	\end{align}
	which gives the angular representations. 
	Based on the limit relation \eqref{p2empdist}, $\tilde{H}_{\bsx,{n}/{k}}$ in equation (9.32) in \citet{ResnickHeavy} serves as a consistent estimator for the standardised angular measure $\tilde{H}_{\bsx}$, where
	$$\tilde{H}_{\bsx,\frac{n}{k}} (\cdot) = \frac{\sum_{\ell=1}^{n}\mathds{1}{\{(R_\ell/b_{\frac{n}{k}},{\boldsymbol{\omega}}_\ell) \in [1,\infty] \times \cdot\}}}{\sum_{\ell=1}^{n}\mathds{1}{\{R_\ell/b_{\frac{n}{k}}\geq1\}}}\stw \tilde{H}_{\bsx} (\cdot),\quad k, n\to\infty\quad \text{and}\quad k/n\to 0.$$
	Since $R^{(k)}/b_{\frac{n}{k}}\stp 1$ \citep[discussion below  (9.32)]{ResnickHeavy}, with $R^{(k)}$ the $k$-th largest radius among $R_1,\ldots,R_n$, replacing ${b}_{{n}/{k}}$ with $R^{(k)}$ implies $\sum_{\ell=1}^{n}\mathds{1} {\{R_\ell\geq R^{(k)}\}}=k$, so
	$$\tilde{H}_{\bsx, n/k} (\cdot) = \frac{1}{k}\sum_{\ell=1}^{n} \mathds{1}{\{R_\ell\geq R^{(k)}, \boldsymbol{\omega}_\ell\in\cdot \}}.
	$$
	We then construct for $\E_{\tilde{H}_{\bsx}} [f(\bs{\omega})]$ the estimator
	\begin{align}\label{p2specemp}
	\hat{\mathbb{E}}_{\tilde{H}_{\boldsymbol{X}}}[f({\boldsymbol{\omega}})]=\frac{1}{k}\sum_{\ell=1}^{n}f({\boldsymbol{\omega}}_\ell)\mathds{1}{\{R_\ell \geq R^{(k)} \}}.
	\end{align}
	
	\subsection{Estimation of the scaling}
	
	To estimate the squared scalings ${\sigma}_{M_S}^2$ of $M_{S}$ for $S\subseteq\{1,\dots,{d}\}$, \citet{KK} apply \eqref{p2specemp} with
	 $f:\Theta_+^{{d}-1}\to\R_+$ defined via the continuous function
	$f(\boldsymbol{\omega})=d\left({\bigvee}_{r \in S}\omega_{ r}^2\right)$.
	In contrast, we only employ the angular measure of the components of $\bs X$ in $S$, which is of cardinality smaller than $d$.
	
	Given the independent replicates ${\boldsymbol{X}}_{1},\dots,{\boldsymbol{X}}_{n}$ of the vector ${\boldsymbol{X}}\in {\rm RV}^{d}_+(2)$ and an ordered vector of indices $I=(I_1,\ldots,I_{|I|})\subsetneq V$ such that $i,j\notin I$, we define
	\begin{enumerate}
	\item[-]\, $\tilde{\boldsymbol{X}}= (X_i, X_j, \bs X_I)\in {\rm RV}^{|I|+2}_+(2)$, and the independent replicates $\tilde{\boldsymbol{X}}_{1},\dots,\tilde{\boldsymbol{X}}_{n}$ with  angular representations $(\tilde{R}_{\ell}, \tilde{\boldsymbol{\omega}}_{\ell})$ for $\ell\in\{1,\dots,n\}$;
	\item[-]\, $\tilde{\boldsymbol{X}}_a= (X_i, aX_j, a\bs X_I)\in {\rm RV}^{|I|+2}_+(2)$, and the independent replicates $\tilde{\boldsymbol{X}}_{a,1},\dots,\tilde{\boldsymbol{X}}_{a,n}$ with angular representations $(\tilde{R}_{a,\ell}, \tilde{\boldsymbol{\omega}}_{a,\ell})$ for  $\ell\in\{1,\dots,n\}$;
		\end{enumerate}
	\begin{enumerate}
		\item[-]\,$\hat{\sigma}_{M_{i,aj,aI}}^2$:
	for a threshold $1\le k\le n$, we employ the non-parametric scaling estimators from the angular measure of $\tilde{\boldsymbol{X}}_{a}$  to estimate ${\sigma}_{M_{i,aj,aI}}^2$ as 
	\begin{align}\label{p2estscal}
	\hat{\sigma}_{M_{i,aj,aI}}^2 
	=\frac{{1+a^2(|I|+1)}}{k}\sum_{\ell=1}^{n}\underset{r\in\{1,\ldots,|I|+2\}}{\bigvee}{\tilde \omega}_{a,\ell r}^2\mathds{1}{\left\{\tilde{R}_{a,\ell} \geq \tilde{R}_a^{(k)}\right\}};
	\end{align}
	
	\item[-]\, $\hat{\sigma}_{M_{i,j,I}}^2$: For $\hat{\sigma}_{M_{i,j,I}}^2$ we use the angular measure of the vector $\tilde{\bs X}_a$, but with the modification
		\begin{align}\label{p2estscala}
		\hat{\sigma}_{M_{i,j,I}}^2 
		=\frac{{1+a^2(|I|+1)}}{k}\sum_{\ell=1}^{n}\underset{r\in\{1,\ldots,|I|+2\}}{\bigvee}{\omega_{\ell r}^2}\mathds{1}{\left\{\tilde{R}_{a,\ell} \geq \tilde{R}_a^{(k)}\right\}},
		\end{align}
			where $\omega_{\ell r}= \tilde{\omega}_{a,\ell r}/a$, for $r\in\{2,\ldots,|I|+2\}$, and $\omega_{\ell 1}=\tilde{\omega}_{a, \ell 1}$. 
	\end{enumerate}	
			Concerning the last point, the angular measure $\tilde{H}_{\tilde{\bs X}_a}$ has atoms $\tilde{\boldsymbol{a}}_k$ with entries 
			$\tilde{a}_{1k}={a}_{ik}$, $\tilde{a}_{2k}=a\,{a}_{jk}$ and  $\tilde{a}_{r+2,k}=a\,{a}_{I_rk}$ for $r\geq1$.
			As ${H}_{\tilde{\bs X}_a}(\Theta^{|I|+1})=1+a^2(|I|+1)$ by \eqref{p2specmass}, the estimator \eqref{p2estscala} estimates the theoretical quantity
			\begin{align*}
			(1+a^2(|I|+1))\mathbb{E}_{\tilde{H}_{\tilde{\bs X}_a}}\left(\underset{r\in\{1,\ldots,|I|+2\}}{\bigvee}{\omega_{\ell r}^2}\right)&=\sum_{k=1}^{d}\norm{\tilde{\bs a}_k}^2\underset{r\in\{1,\ldots,|I|+2\}}{\bigvee}\frac{{a}_{rk}^2}{\norm{\tilde{\boldsymbol{a}}_k}^2}=\sum_{k=1}^{d}\underset{r\in\{1,\ldots,|I|+2\}}{\bigvee}{a}_{rk}^2={\sigma}_{M_{i,j,I}}^2,
			\end{align*}
			which can be used to establish consistency of $\hat{\sigma}_{M_{i,j,I}}^2$, as shown in Section \ref{p2consistency}.
Alternatively, one can estimate ${\sigma}_{M_{i,j,I}}^2$ directly from the angular measure of $\tilde{\boldsymbol{X}}$, using 
		\begin{align}\label{p2estscalinit}
		\hat{\sigma}_{M_{i,j,I}}^2 
		=\frac{{2+|I|}}{k}\sum_{\ell=1}^{n}\underset{r\in\{1,\ldots,|I|+2\}}{\bigvee}{\omega_{\ell r}^2}\mathds{1}{\left\{\tilde{R}_{\ell} \geq \tilde{R}^{(k)}\right\}}.
		\end{align}
	The estimators \eqref{p2estscalinit} and  \eqref{p2estscala} perform similarly in the inital step of the algorithm, when $I=\emptyset$. However, for $I\neq \emptyset$, due to the increasing dimension of the angular measure, we use only \eqref{p2estscal} and \eqref{p2estscala} to estimate the scalings involved in the matrix ${\Delta}_{\hat I}$ in Algorithm~\ref{p2causordalg}.
		
	\subsection{Consistency of the scaling estimates}	\label{p2consistency}	
	Results from \citet[Theorem 1]{lars}  were used in \citet[Theorem 4]{KK} to show consistency and asymptotic normality of the estimators in \eqref{p2estscal}, \eqref{p2estscala}, and \eqref{p2estscalinit}, under a mild condition.		
	We use the angular representation \eqref{p2pol.est}, and focus on the general form of the estimator $\hat{\mathbb{E}}_{\tilde{H}_{\boldsymbol{X}}}[f({\boldsymbol{\omega}})]$ in \eqref{p2specemp}.
	
	Let the radial component $R$ of $\bsx$ have distribution function $F$. We restate a version of Theorem 4 from \citet{KK}. 
	\begin{theorem}[Central Limit Theorem]\label{p2clt}
		Let $\boldsymbol{X}\in {\rm RV}^{d}_+(2)$ with angular representation $(R,\bs\omega)$
		and let $\boldsymbol{X}_1,\dots,\boldsymbol{X}_n$ be independent replicates of $\boldsymbol{X}$. Choose $k$ such that $k=o(n)$ and $k\to\infty$ as $n\to\infty$. Let $f\colon \Theta_+^{{d}-1}\to {\mathbb{R}_+}$ be a continuous function. 
		Assume that 
		\begin{align}
		\lim\limits_{n\to \infty}\sqrt{k}\left(\frac{n}{k}\mathbb{E}\left[f({{\boldsymbol{\omega}_1}})\mathds{1}{\left\{R_1\geq b_{\frac{n}{k}}t^{-1/ \alpha}\right\}}\right] 
		- \mathbb{E}_{\tilde{H}_{\bsx}}[f({\boldsymbol{\omega}}_1)]\frac{n}{k}\bar{F}\left(b_{\frac{n}{k}} t^{-1/\alpha}\right) \right)=0
		\label{p2assump}
		\end{align}
		holds locally uniformly for $t\in[0, \infty)$, and
		$\sigma^2\coloneqq\emph{Var}_{\tilde{H}_{\bsx}}(f({\boldsymbol{\omega}}))>0$.
		Then
		\begin{align*}
		\sqrt{k}\left(\hat{\mathbb{E}}_{\tilde{H}_{\boldsymbol{X}}}\left[f({\boldsymbol{\omega}})\right]-\mathds{E}_{\tilde{H}_{\boldsymbol{X}}}[f({\boldsymbol{\omega}})]\right)\overset{\D}{\to} \mathcal{N}(0, \sigma^2), \hspace{5mm} n\to\infty.
		\end{align*}
	\end{theorem}		
	
	We apply the theorem  to the vector $\tilde{\bs X}_a=(X_i, aX_j, aX_I)$ for $I=(I_1,\ldots, I_{|I|})\subsetneq V$ with $i,j\notin I$  and for the functions $f$ used in the estimators \eqref{p2estscal}, \eqref{p2estscala} and \eqref{p2estscalinit}. This gives consistency of the estimators, similar to \citet[Theorem 5]{KK}, which we adapt below. 
	
	\begin{theorem}
		\label{p2scalclt}
		Let $\boldsymbol{X}=A\boldsymbol{Z}$ be an LSEM+ with innovation coefficient matrix $A$ and let $\boldsymbol{X}_1,\dots,\boldsymbol{X}_n$ be independent replicates of $\boldsymbol{X}$. Take the subvector $\tilde{\bs X}_a$  for $I\subsetneq \{1,\dots,{d}\}$ with $i,j\notin I$ and consider the independent replicates $\tilde{\boldsymbol{X}}_{a,1},\dots,\tilde{\boldsymbol{X}}_{a,n}$ with angular representations $(\tilde{R}_{a,\ell}, \tilde{\boldsymbol{\omega}}_{a,\ell})$ for  $\ell\in\{1,\dots,n\}$. Choose $k$ such that $k=o(n)$ and $k\to\infty$ as $n\to\infty$.
		Furthermore, assume that (\ref{p2assump}) holds for $\tilde{\bs X}_a$ and $f\left(\boldsymbol{\tilde\omega_a}\right)={\left(1+a^2(|I|+1)\right)}\left({\vee}_{r\in\{1,\ldots,|I|+2\}}\tilde\omega_{a,r}^2\right)$ and that 
		$\emph{Var}_{\tilde{H}_{\tilde\bsx_a}}(f(\tilde{\bs\omega}_a))>0$. 
		Then
		$$	\sqrt{k}\left(\hat{\mathbb{E}}_{\tilde{H}_{\tilde{\bsx}_{a}}}\left[f\left({\tilde{\bs\omega}_a}\right)\right]-\mathds{E}_{\tilde{H}_{\tilde{\bsx}_{a}}}\left[\left({\tilde{\bs\omega}_a}\right)\right]\right)\overset{\mathcal{D}}{\to}  {\mathcal{N}}(0, v^2),\quad n\to\infty,$$ 
		with variance
		\begin{align*}
		v^2
		 = {\left(1+a^2(|I|+1)\right)}\left(\sum_{p=1}^{{d}}\underset{r\in\{1,\ldots,|I|+2\}}{\bigvee}\frac{\tilde a_{rp}^4}{\norm {\tilde{\bs a}_p}^2_2}\right)-
		\left(\sum_{p=1}^{d}\underset{r\in\{1,\ldots,|I|+2\}}{\bigvee}\tilde a_{rp}^2\right)^2,
		\end{align*}
		where the vector $\tilde{\bs{a}}_p\in \mathbb{R}^{2+|I|}_+$ has entries 
		$\tilde{a}_{1k}={a}_{ik}$, $\tilde{a}_{2k}=a\,{a}_{jk}$ and  $\tilde{a}_{r+2,k}=a\,{a}_{I_rk}$ for $r\geq1$.
		
	\end{theorem}

\subsection{Proof of consistency of Algorithm \ref{p2causordalg}}		
\begin{proof}[Proof of Proposition \ref{p2consistencyy}]
	Let $S$ be the total number of iterations within the \textbf{while} loop of Algorithm \ref{p2causordalg}, and note that $S\leq d$, because $I$ is augmented by at least one element at each iteration. For each  iteration $s$ within the while loop, we let $\pi^s$ denote the output of the vector $\pi$ in line 6 of Algorithm~\ref{p2causordalg}, and write $I^s=(\pi^s, \pi^{s-1},\ldots, \pi^1)$ for the ordered vector of nodes estimated up to iteration~$s$. Likewise, we use $\hat{\Delta}^s$ to denote the matrix ${\Delta}_{\hat I^s}$ in line 3.
	We want to show that the probability of an incorrect causal ordering converges to zero, namely, 
	$$\mathbb{P}(\exists i,j\in V, j\in\an(i):I_j< I_i )\to 0,\quad n\to\infty.$$
	Clearly, $I=I^S$ is an incorrect causal ordering if only if there is an incorrect ordering, or error, in the $s$-th step of the algorithm, for some $s\leq S$, that is, if  $\pi^s_j<\pi^s_i$ for $i,j\in V\setminus I^s$ such that $j\in\an(i)$.
	
	We proceed via induction over the number of steps, $s$,  and show that the probability of obtaining an error at a particular step converges to zero.
	
	Consider the probability of an error at step $s=1$. Due to the existence of a causal ordering of the nodes supported on a DAG, Algorithm \ref{p2causordalg} should select only amongst the source nodes in this step. Theorem \ref{p2sourcenodes} establishes that $\Delta_{ij}^s=0$ for all nodes $j\in V_0$ and $j\in \an(i)$, so $\hat{\Delta}^s_{ij}\overset{P}{\to}0$ by consistency of the estimates, whereas  $\Delta_{ji}^s=c_{ji}<0$, which implies that $\hat{\Delta}^s_{ji}\overset{P}{\to}c_{ji}$. The continuous mapping theorem applied to the minimum function gives that $\min_{i \in V\setminus\{j\}}\hat{\Delta}^s_{ij}\overset{P}{\to}0$. By the existence of source nodes, it follows that the error term $\eps_{\hat I^s}$ converges in probability to zero, i.e.,  $\eps_{\hat I^s}\overset{P}{\to}0$. 
	Consider the event
	$$E^s\coloneqq \{\textrm{select node $i\in V\setminus I^s$, such that $\exists j\in \an(i)\cap V\setminus I^s$, and $\pi^s_i> \pi^s_j$}\}.
	$$
	Then, the probability of an error at step $s=1$ equals
	\begin{align}
	\mathbb{P}(E^s)&\leq \mathbb{P}\left( \underset{\left\{i,j\in V: j\in \an(i)\right\}}\cup \left\{\min_{j \in V\setminus\left\{i\right\}}\hat{\Delta}^s_{ji}> \min_{i \in V\setminus\left\{j\right\}}\hat{\Delta}^s_{ij}\right\}\cap\left\{ |{ \delta}_{\hat I^s,i}|\leq   \eps_{\hat I^s}\right\}\right) \nonumber \\
	& \leq  \sum_{\{i,j\in V: j\in \an(i)\}}\mathbb{P}\left(\left\{ |{ \delta}_{\hat I^s,i}|\leq  \eps_{\hat I^s}\right\}\right)\overset{P}{\to}{0},\quad n\to\infty, \label{p2step_s1}
	\end{align}
	because ${ \delta}_{\hat I^s,i}\overset{P}{\to}c_i<0$.  
	Hence, the probability of selecting a node $i$ with $\an(i)\neq \emptyset$ at step $s=1$ converges to zero.
	
	Suppose, by the inductive hypothesis, that we can consistently recover a causal ordering up to step $r-1\leq S-1$, that is,
	\begin{align}\label{p2inductohyp}
	\mathbb{P}\left(\underset{s\leq r-1}{\cup}E^s\right)\overset{P}{\to}{0},\quad n\to\infty.
	\end{align}
	
	Let $s=r$ and consider the event $E^r$, whose probability, given the recursive nature of Algorithm \ref{p2causordalg}, we write as
	\begin{align}\label{p2Edecomp}
\mathbb{P}(E^r)&=	\mathbb{P}\left(E^r\cap  \underset{s_1\leq r-1}{\cup}E^{s_1}\right)+\mathbb{P}\left(E^r \cap \left\{\underset{s_1\leq r-1}{\cup}E^{s_1}\right\}^c\right),\nonumber\\
	&=\mathbb{P}\left(E^r\cap  \underset{{s_1}\leq r-1}{\cup}E^{s_1}\right)+\mathbb{P}\left(E^r | \left\{\underset{{s_1}\leq r-1}{\cup}E^{s_1}\right\}^c\right)\mathbb{P}\left(\left\{\underset{{s_1}\leq r-1}{\cup}E^{s_1}\right\}^c\right).
	\end{align}
	The first term in the last line of \eqref{p2Edecomp} vanishes as $n\to\infty$ because of the inductive hypothesis \eqref{p2inductohyp}, i.e., $\mathbb{P}({\cup}_{{s_1}\leq r-1}E^{s_1})\overset{P}{\to}{0}$.
	The conditioning event $\{{\cup}_{{s_1}\leq r-1}E^{s_1}\}^c$ in the second term on the right-hand side of \eqref{p2Edecomp} satisfies the setting in Theorem \ref{p2descnodes}, that there are no ancestors of $I^s$ outside the set $I^s$, and, by \eqref{p2inductohyp}, occurs with probability approaching one. Due to the existence of a causal ordering of the nodes on a DAG, at every step $s$ we can identify a node $j_s$ with no ancestors in $V\setminus I^s$. Theorem \ref{p2descnodes} and consistency of the estimators imply that for such a node $\min_{i \in V\setminus I^s\cup \{j_s\}}\hat{\Delta}^s_{ij_s}\overset{P}{\to}0$, and therefore $\eps_{\hat I^s}\overset{P}{\to}0$ for $s=r$.
	
	Now consider the two nodes $i,j\in V\setminus I^r$, with $j\in\an(i)$.  
	By Theorem \ref{p2descnodes}, $\hat{\Delta}^r_{ji}\overset{P}{\to}c_{ji}<0$, which implies that $\min_{j \in V\setminus\{i\}}\hat{\Delta}^r_{ji}={ \delta}_{\hat I^r,i}\overset{P}{\to}c_{i}<0$.
	We consider $\mathbb{P}(E^r | \{{\cup}_{{s_1}\leq r-1}E^{s_1}\}^c)$ and, similar to \eqref{p2step_s1}, we apply the upper bound
	\begin{align*}
	\mathbb{P}\left(E^r | \left\{\underset{{s_1}\leq r-1}{\cup}E^{s_1}\right\}^c\right)\leq  \sum_{\left\{\,i,j\,\in\, V\setminus I^r \,:\, j\,\in\, \an(i)\right\}}\mathbb{P}\left(\left\{ |{ \delta}_{\hat I^r,i}|\leq  \eps_{\hat I^{s_1}}\right\} | \left\{\underset{{s_1}\leq r-1}{\cup}E^{s_1}\right\}^c\right)\overset{P}{\to}{0}, \quad n\to\infty.
	\end{align*}
	
	This proves the inductive claim, that we consistently recover a correct ordering at every step of Algorithm \ref{p2causordalg}.
\end{proof}

\subsection{Standardisation of the margins}

The methodology of this paper requires that the vector $\boldsymbol{X}\in {\rm RV}^{d}_+(2)$ has unit scalings. In practice, one  may not know the index of regular variation $\alpha$ needed to standardise the margins. In both, the simulation study, where $\alpha\in\{2, 3\}$, and the data analysis, we transform the data to Fr\'echet margins using the empirical probability integral transform \citep[p. 338]{beirlant}.  Given original observations $\bs X^*_1,\ldots, \bs X^*_n$, we obtain the standardised margins~via
\begin{align}\label{p2empstand}
X_{\ell i}\coloneqq\Big\{-\log\Big(\frac{1}{n+1}\sum_{j=1}^{n}\mathds{1}_{\{X^*_{ji}\leq X^*_{\ell i}\}}\Big)\Big\}^{-1/2},\quad  \ell\in\{1,\dots,n\}.
\end{align}

EASE also applies the empirical integral transform to uniform margins, so this standardisation step does not affect its performance.

\newpage
	\section{Simulation results}\label{p2bplot5000}
	This section contains the results of the simulation study for the sample size $n=5000$, commented but not shown in Section~\ref{p2simres}.

	\begin{figure}[H]\centering
		\hspace{-0.4cm}\includegraphics[ height=5cm, width=13cm, clip]{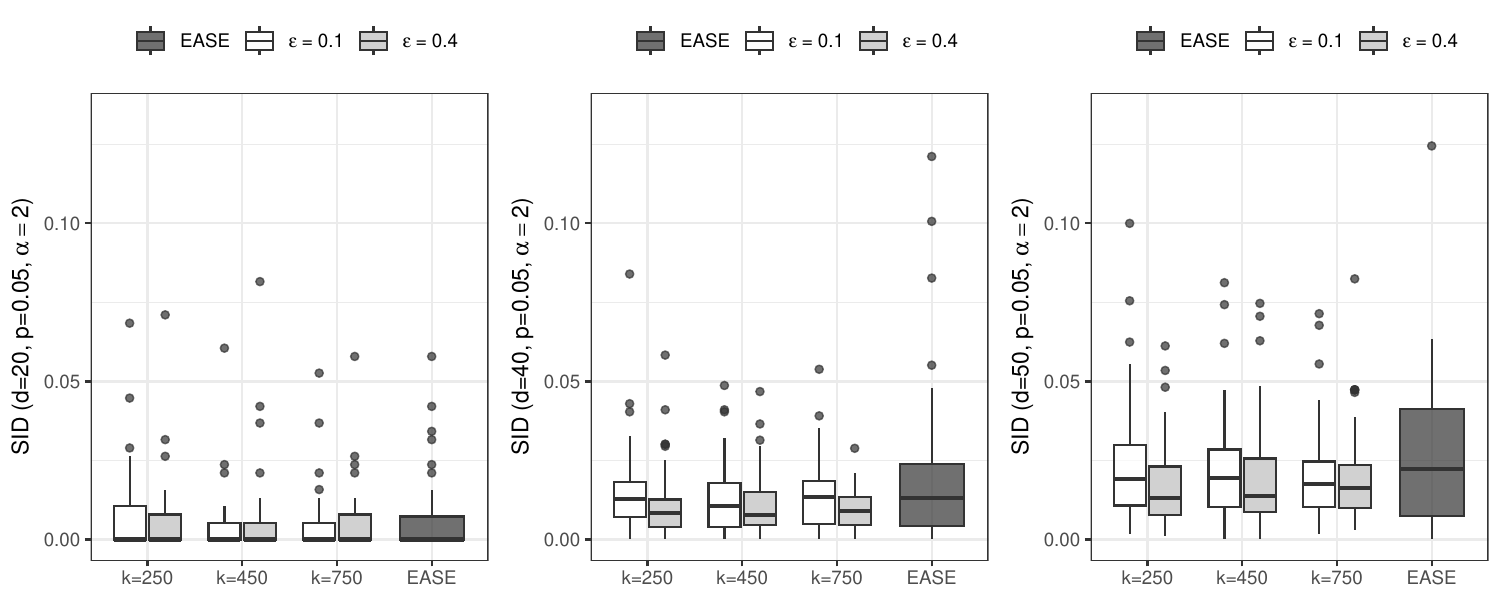}
		
		\hspace{-0.4cm}\includegraphics[ height=4.3cm, width=13cm, clip]{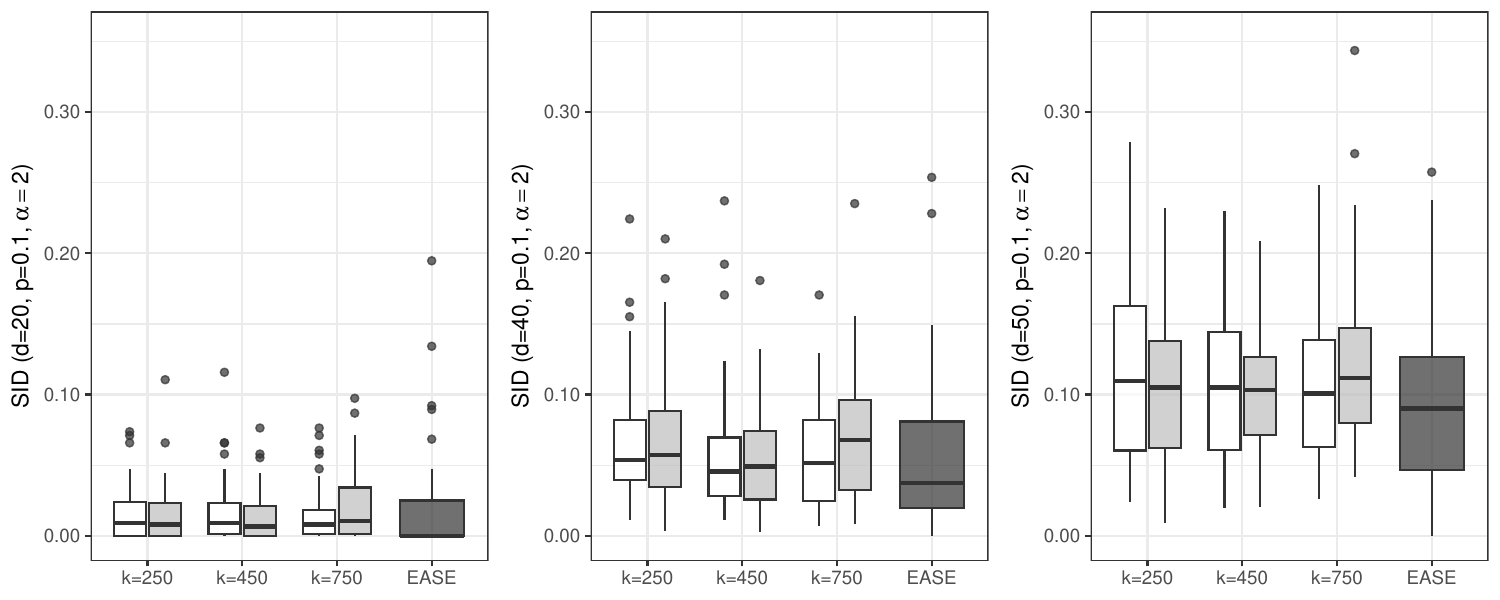}
		
		\hspace{-0.4cm}\includegraphics[ height=4.3cm, width=13cm, clip]{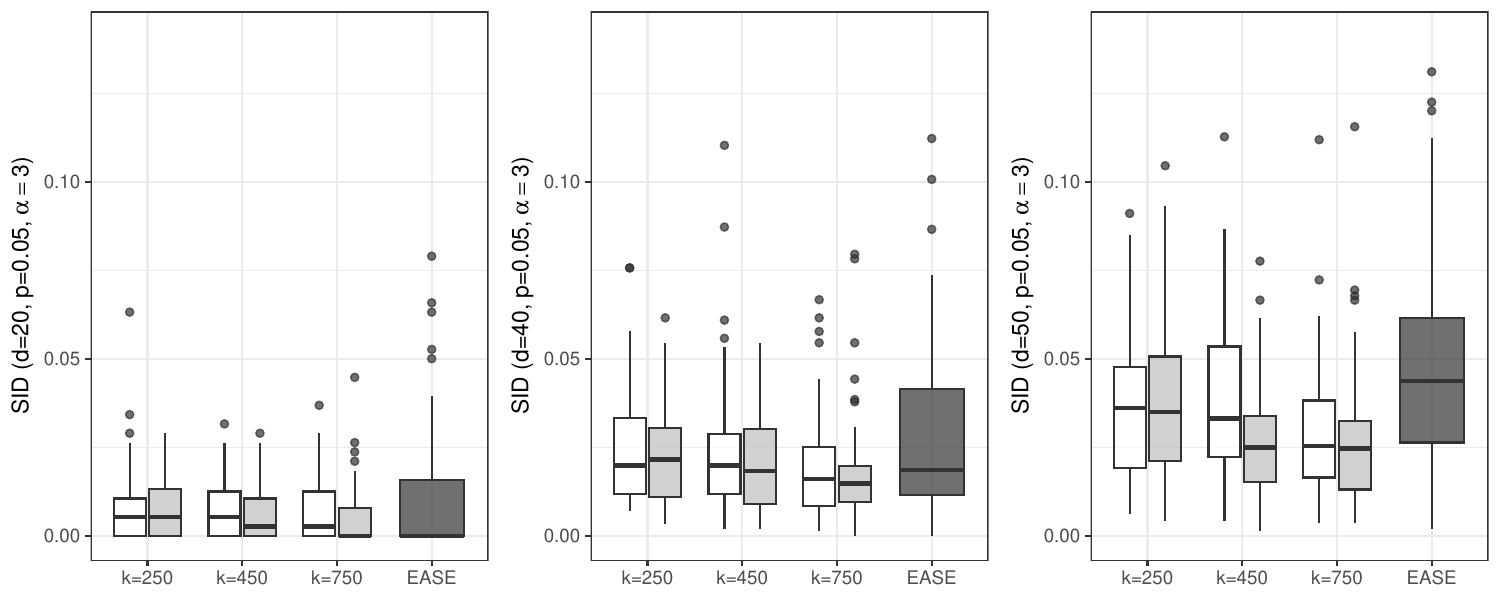}
		
		\hspace{-0.4cm}\includegraphics[ height=4.3cm, width=13cm, clip]{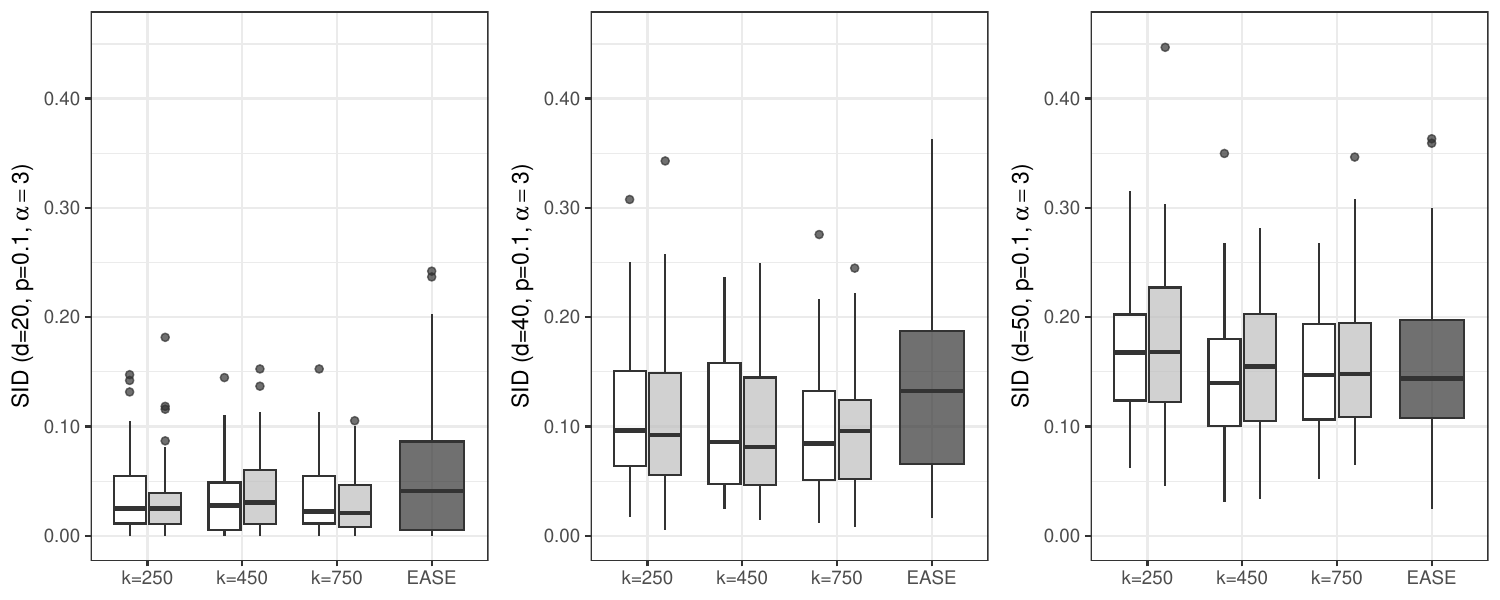}
		
		\caption{Boxplots of the SID of the causal orderings of Algorithm \ref{p2causordalg} and of EASE for 50 random DAGs for each configuration of $(d, p, \alpha)$ and for $n=5000$.} \label{p2runs5000alt}
	\end{figure}

	\begin{figure}[H]\centering
		\hspace{-0cm}\includegraphics[ height=4cm, width=12cm, clip]{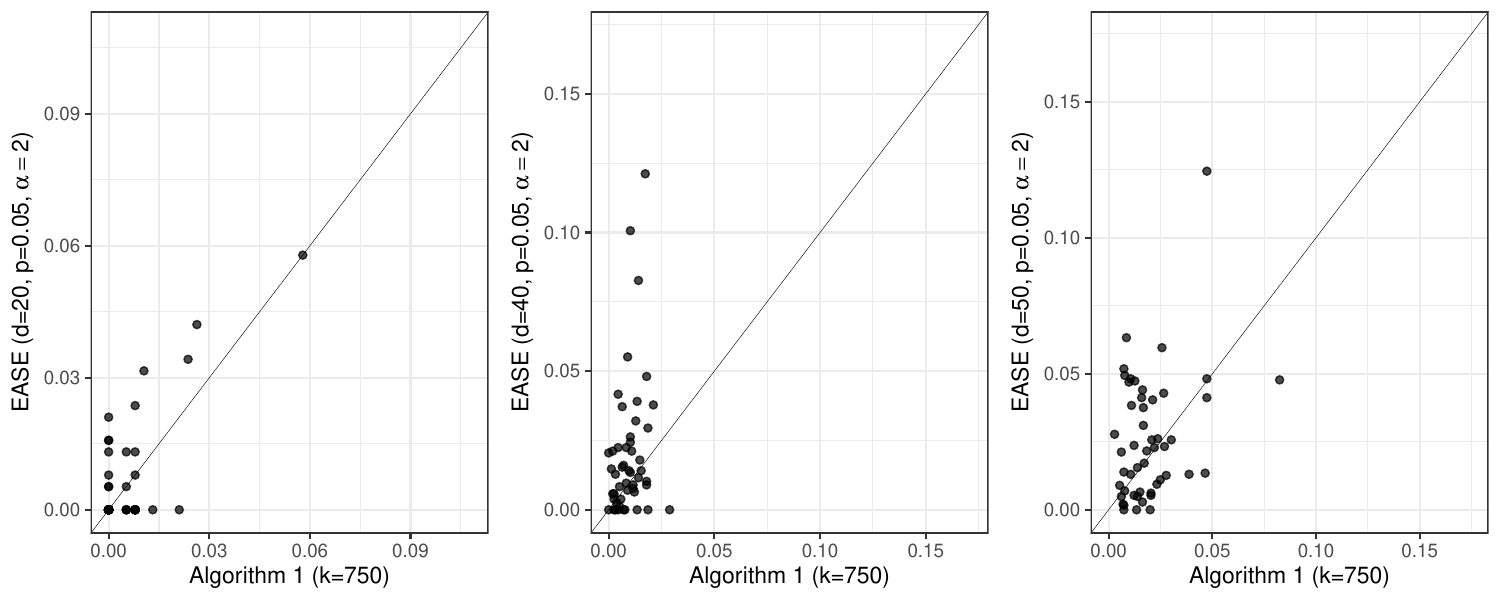}
		
		\hspace{-0cm}\includegraphics[ height=4cm, width=12cm, clip]{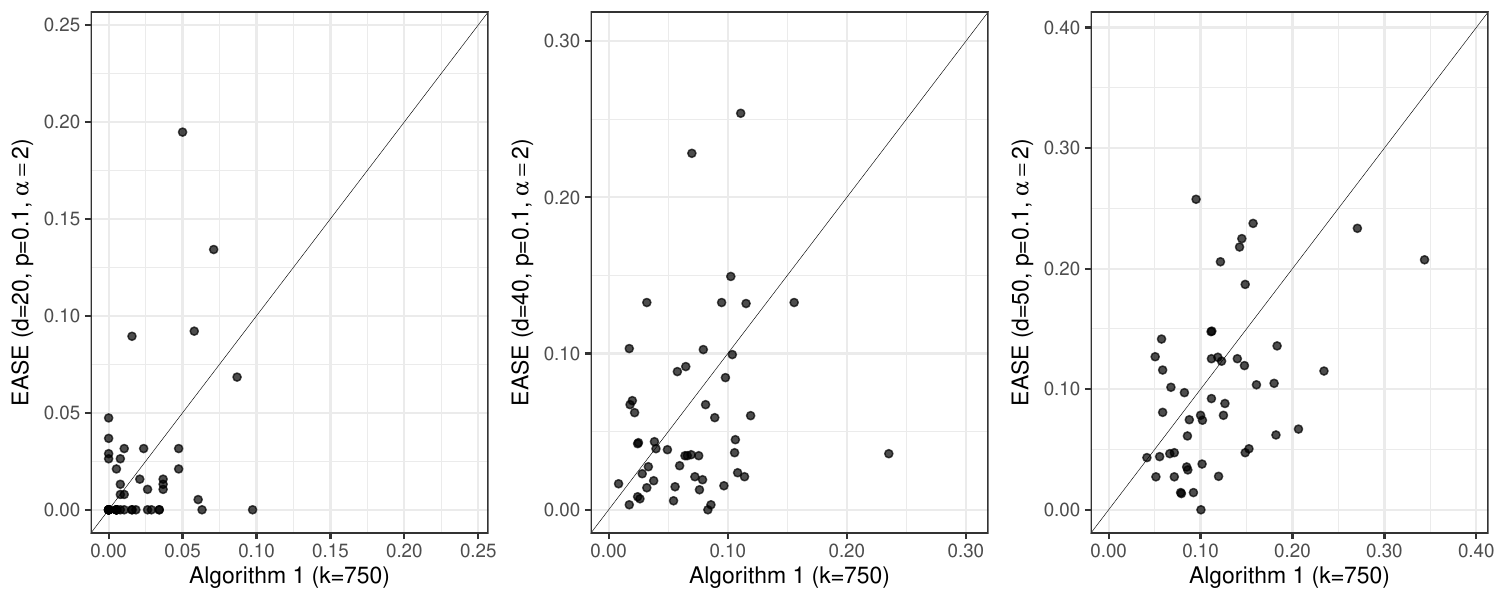}
		
		\hspace{-0cm}\includegraphics[ height=4cm, width=12cm, clip]{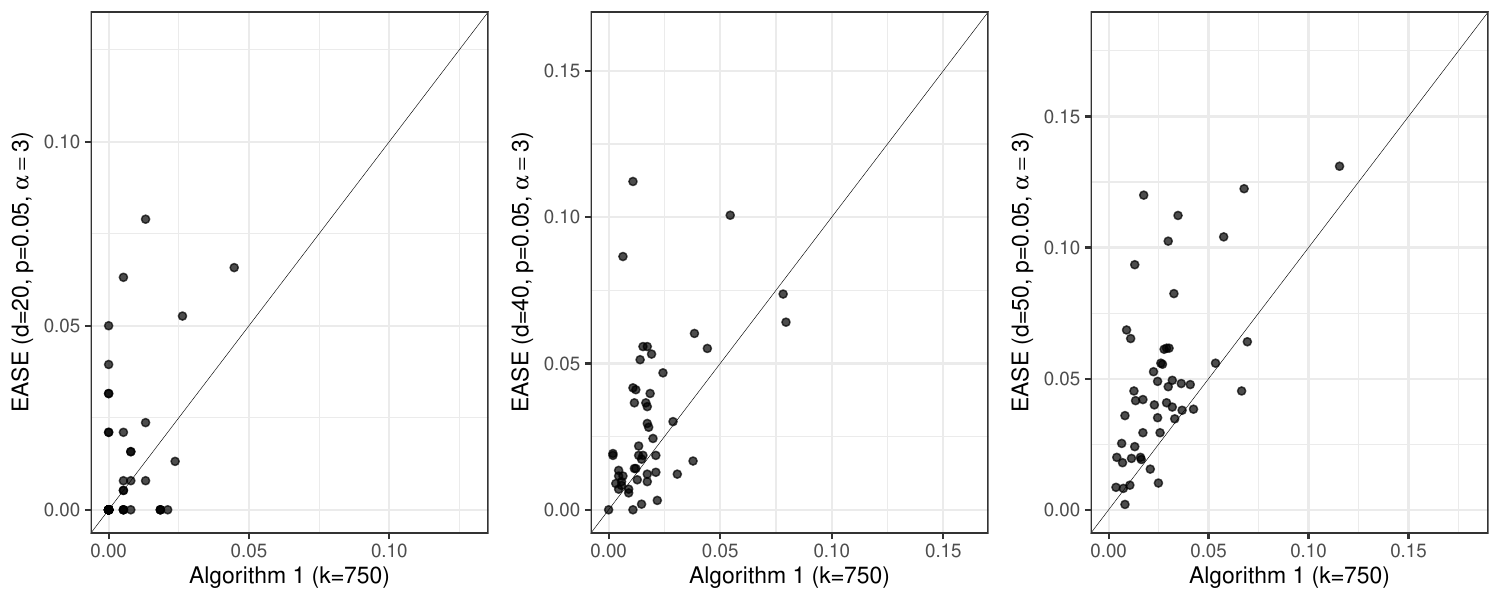}
		
		\hspace{-0cm}\includegraphics[ height=4cm, width=12cm, clip]{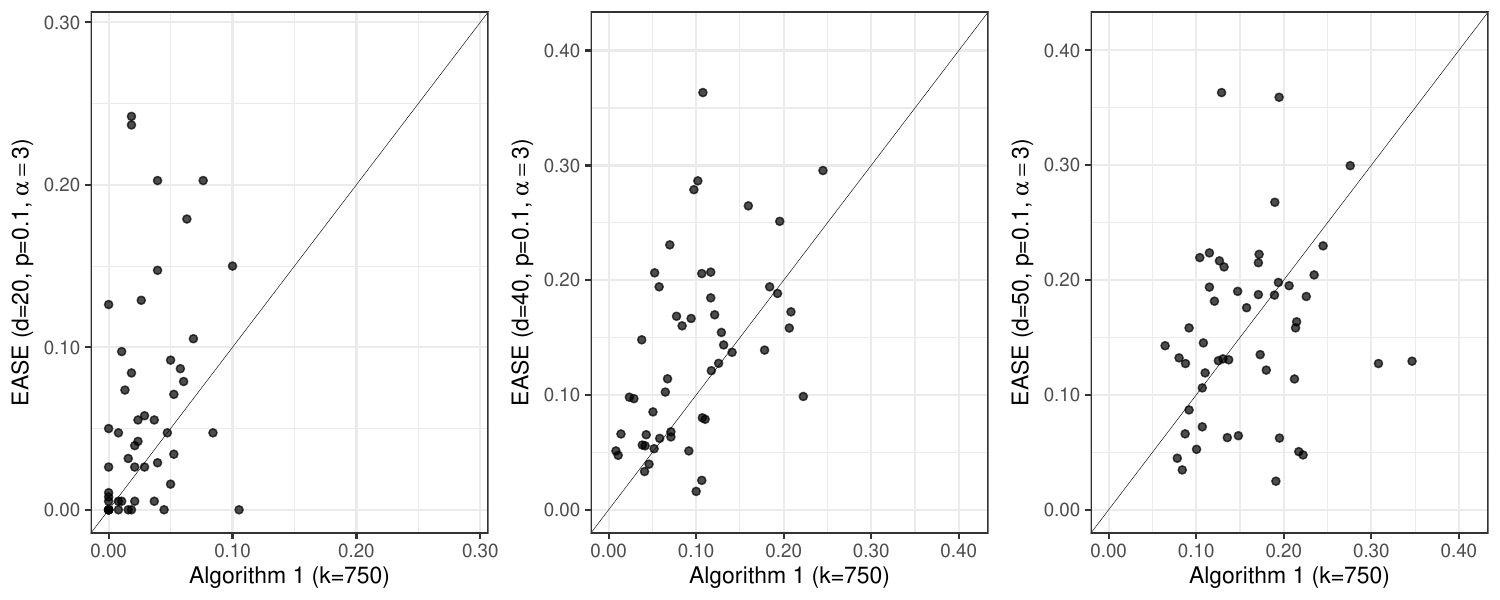}
		
		\caption{Scatterplots of the SID of the causal orderings of Algorithm~\ref{p2causordalg} and of EASE for 50 random DAGs for each configuration of $(d, p, \alpha)$ and for $n=5000$.} \label{p2vsruns5000}
	\end{figure}
	
	\subsection{Choice of the scalar $a$}\label{p2ascal}
	
	In this section we investigate the stability of Algorithm~\ref{p2causordalg} for the scalar $a$ for the smaller sample size $n=1000$.  Following the recipe outlined in Section \ref{p2simsetup}, we perform 20 simulations over random DAGs from configurations $(d, p, \alpha)$, where each parameter is randomly sampled from $d\in\{20, 40, 50\}$, $p\in\{0.05, 0.1\}$ and $\alpha\in\{2, 3\}$. To each DAG we apply Algorithm~\ref{p2causordalg} for each scalar $a\in\{1.0001, 1.15, 1.3, 1.5,2\}$ and estimate a causal ordering. In this simulation we fix $\eps=0.4$ and $k=250$. The resulting plots in Figure~\ref{p2sca} indicate better performance for values of $a$ larger than $1.0001$,  with $a=1.3$ giving the lowest SID when averaged over all DAGs. For the denser graphs with $p=0.1$ the same choice of $a$ yields a better and less variable performance.

	\begin{figure}[H]\centering
		\hspace{0cm}\includegraphics[ height=10cm, width=12cm, clip]{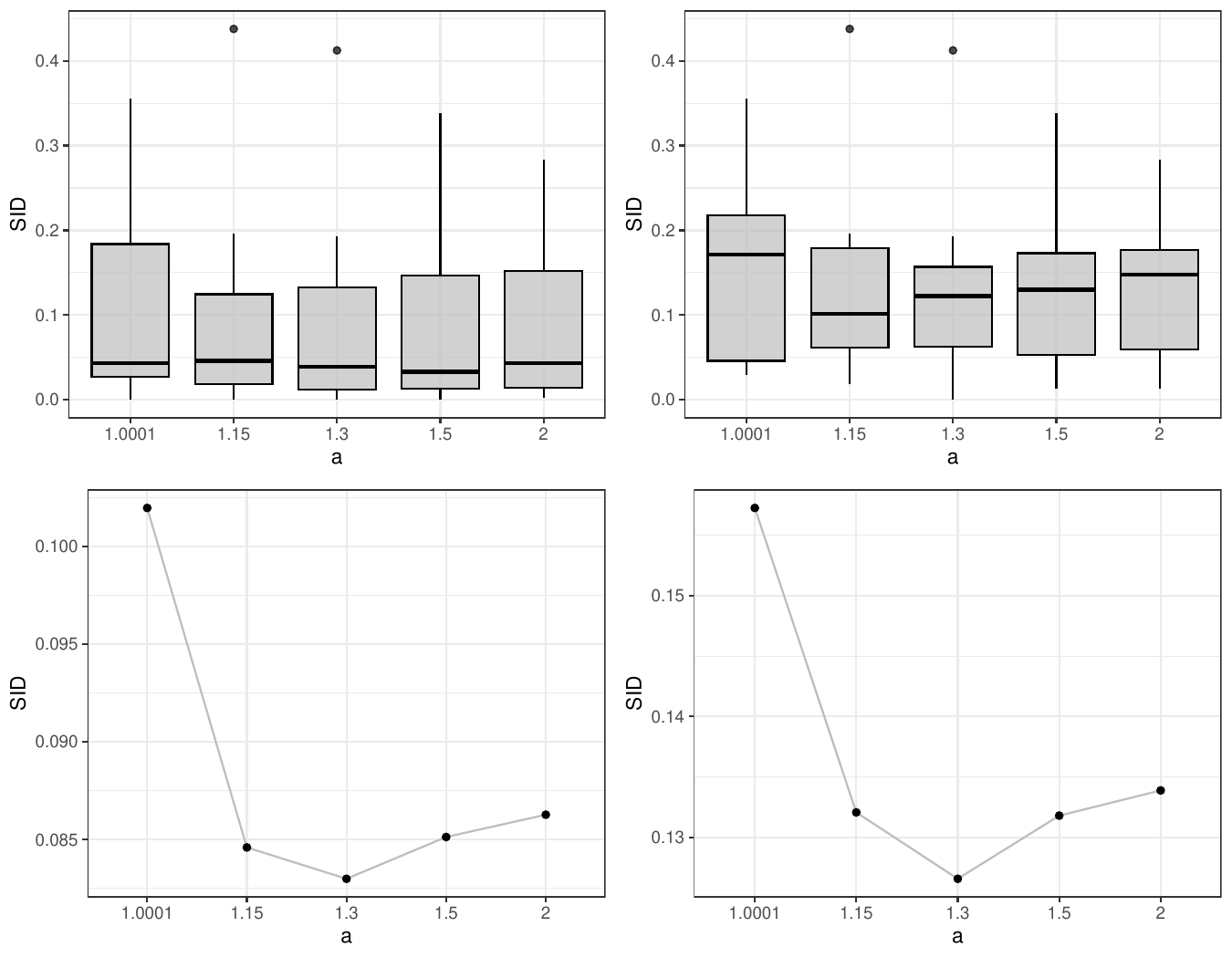}
		
		
		\caption{Boxplots (top) and mean plots (bottom) of the SID of the causal orderings of Algorithm~\ref{p2causordalg} for 20 random DAGs for different choices of $a$. The left-hand panels show the results over all DAG configurations, and the right-hand panels show the results only for the denser DAGs  ($p=0.1$).} \label{p2sca}
	\end{figure}

	\newpage

\end{appendices}

\end{document}